\documentclass{article}

\PassOptionsToPackage{numbers, compress}{natbib}
\usepackage[preprint]{neurips_2026}
\makeatletter
\renewcommand{\@noticestring}{}
\makeatother

\usepackage[utf8]{inputenc} 
\usepackage[T1]{fontenc}    
\usepackage{hyperref}       
\usepackage{url}            
\makeatletter
\g@addto@macro{\UrlBreaks}{\UrlOrds}   
\makeatother
\usepackage{booktabs}       
\usepackage{amsfonts}       
\usepackage{nicefrac}       
\usepackage{microtype}      
\usepackage{xcolor}         
\usepackage{amsmath}        
\usepackage{amsthm}         
\usepackage{graphicx}       
\usepackage{subcaption}     
\usepackage{tikz}           
\usepackage{enumitem}       
\usepackage{rotating}       
\usepackage{comment}        
\usepackage{multirow, makecell}
\usepackage{tabularx}

\newtheorem{theorem}{Theorem}
\newtheorem{lemma}[theorem]{Lemma}
\newtheorem{remark}[theorem]{Remark}
\newcommand{\Fsend}{F_{\mathrm{s}}^{\mathrm{send}}}

\title{Don't Let a Few Network Failures Slow the Entire AllReduce}

\author{%
  \bfseries Peiqing Chen$^{1}$ \quad
  \bfseries Jiedong Jiang$^{2}$ \quad
  \bfseries Nengneng Yu$^{1}$ \quad
  \bfseries Yuefeng Wang$^{3}$ \\
  \bfseries Sixian Xiong$^{1}$ \quad
  \bfseries Wei Wang$^{1}$ \quad
  \bfseries Zaoxing Liu$^{1}$ \\[4pt]
  \normalfont
  $^{1}$University of Maryland, College Park \quad
  $^{2}$Utrecht University \quad
  $^{3}$Kyoto University
}

\begin{document}

\maketitle

\begin{abstract}
  Network failures are among the most frequent hardware faults in
  large-scale GPU clusters and a leading cause of training-job
  interruptions. Modern collective communication libraries such as NCCL
  mitigate network failures by rerouting traffic through surviving NICs
  on the same server, trading reduced inter-node bandwidth for
  uninterrupted training. However, the degraded server remains on the
  critical path of the standard ring algorithm, slowing the \emph{entire}
  collective. We present the first information-theoretic lower bound on
  AllReduce completion time under asymmetric network bandwidth and show that
  when the straggler retains at least half of its original bandwidth,
  the unavoidable overhead relative to the fault-free optimum is only
  $O(1/p)$ for $p$~GPUs. We then design \textsc{OptCC}, a four-stage
  pipelined AllReduce algorithm that approaches this lower bound.
  Experiments on SimAI confirm that \textsc{OptCC} closes the gap left
  by existing fault-tolerant schemes: under practical network failures
  with up to 50\% bandwidth loss, \textsc{OptCC} completes AllReduce
  within 2--6\% of NCCL's fault-free ring performance, whereas the
  state-of-the-art incurs up to 57\% overhead.
\end{abstract}

\section{Introduction}
\label{sec:introduction}

Network failures, including NIC, cable, and port faults, are among the
most frequent hardware faults in large-scale GPU clusters, wasting
10--15\% of GPU hours in production training
runs~\cite{kokolis2025reliability, llama3herd2024, gunawi2018failslow, diskin2021distributed}.
Recent systems mitigate these failures by rerouting the traffic on
these GPUs through surviving NICs on the same server via
PXN~\cite{nccl_pxn2022}, keeping the training job alive without a
full restart~\cite{shift2025, ncclx2025, ryabinin2021moshpit}
(Figure~\ref{fig:nic-failure-recovery}).
However, this reduces the affected server's aggregate inter-node
bandwidth, turning its GPUs into
\emph{stragglers}---GPUs whose effective network bandwidth is lower
than their healthy peers. On mainstream multi-NIC servers such as
those with H100 or B100 GPUs, losing a single NIC cuts bandwidth
by 12.5\%.

The collective most affected by this bandwidth degradation is the
AllReduce used for data-parallel gradient
synchronization~\cite{gibiansky2017allreduce, narayanan2021megatron, rajbhandari2020zero}.
Ring AllReduce~\cite{patarasuk2009bandwidth, thakur2005optimization, nvidia2024nccl} is
bandwidth-optimal on homogeneous links, but this optimality breaks
down on heterogeneous
topologies~\cite{zhao2024forestcoll}; a network failure creates
precisely such asymmetry.
After failover, existing systems resume the standard Ring AllReduce
without adapting to the asymmetric bandwidth, causing congestion at
the straggler's degraded link that throttles the entire collective.
Even the most sophisticated optimization among
these~\cite{r2ccl2025} still incurs up to 57\% increase in AllReduce
runtime when 50\% of the bandwidth is lost on a single straggler.
\emph{Existing approaches leave a significant performance gap between
fault-free AllReduce and the degraded setting; in this paper, we aim to
close it.}

Our key insight is that a straggler's irreducible contribution to
AllReduce is inherently \emph{small}: to guarantee correctness, the
straggler must transmit its own private data at least once, but all
remaining data transfers can be handled entirely by healthy GPUs.
Crucially, these two parts use physically separate links---the
straggler's slow link and the healthy GPUs' fast links---so they can
proceed in parallel.
Since the straggler's bandwidth is only one out of $p$ links in the
cluster, a well-designed pipeline can absorb this loss with an overhead
of only $O(1/p)$ relative to the fault-free optimum.
This insight extends to the more common production scenario of
multiple stragglers, as long as the total bandwidth lost across all
degraded nodes remains sufficiently smaller than the healthy GPUs'
aggregate link bandwidth.

In summary, we make three contributions:
\begin{enumerate}[leftmargin=*,itemsep=1pt,topsep=2pt,parsep=0pt]
  \item \textbf{Information-theoretic lower bound.}
        We prove that the unavoidable overhead of bandwidth degradation is
        surprisingly small (Table~\ref{tab:overhead}): when the straggler
        retains at least half of its original bandwidth, the overhead
        over the fault-free optimum is only $O(1/p)$---less than 1\% at
        $p{=}128$ GPUs.
        We extend the bound to more general multiple-straggler and
        multi-GPU-per-server scenarios (Section~\ref{sec:lower-bound}).

  \item \textbf{Algorithm achieving the lower bound.}
        We design \textsc{OptCC}, a four-stage pipelined AllReduce that
        achieves this lower bound with zero idle bubbles in the
        steady-state pipeline body.
        The schedule is a closed-form construction requiring no
        optimization solver, supporting online generation when a failure
        is detected (Section~\ref{sec:algorithm}).

  \item \textbf{SimAI-based evaluation.}
        On the SimAI\footnote{SimAI is a large-scale, production-level network simulator that models GPU clusters and communication libraries to evaluate distributed training and collective communication under realistic network conditions.} network simulator~\cite{simai2025}, OptCC completes
        AllReduce within 2--6\% of the fault-free NCCL ring under
        practical network failures, whereas existing approaches incur
        up to 57\% overhead (Section~\ref{sec:experiments}).
        Code is available at \url{https://anonymous.4open.science/r/OPTCC_NeurIPS2026-1E87/}.
\end{enumerate}

\begin{table}[t]
\centering
\footnotesize
\caption{Theoretical lower bounds on straggler overhead relative to the fault-free baseline~$T_0$; $p$: number of GPUs; $\ell$ (resp.\ $\ell_i$): straggler slowdown factor; $\ell_1 = \max_i \ell_i$.
For one GPU per server, $T_0 = 2(p{-}1)n/p$;
for $g$ GPUs per server, $T_0 = 2(p{-}1)n/(gp)$.}
\label{tab:overhead}
\vspace{2pt}
\begin{tabular}{@{}lcc@{}}
\toprule
& $\ell\;(\ell_1) < 2 {-} \tfrac{1}{p-1}$ & $\ell\;(\ell_1) \geq 2 {-} \tfrac{1}{p-1}$ \\
\midrule
\textbf{No straggler}
& $T_0$ & $T_0$ \\[4pt]
\textbf{Single straggler, 1 GPU/server in DP group} (Thm.~\ref{thm:lower-bound})
& $\bigl(1 + \tfrac{\ell-1}{\ell\,p} + O\bigl(\tfrac{1}{p^2}\bigr)\bigr) T_0$
& $\bigl(\tfrac{\ell}{2} + \tfrac{\ell}{2\,p} + O\bigl(\tfrac{1}{p^2}\bigr) \bigr)\,T_0$ \\[6pt]
\textbf{$m$ stragglers, 1 GPU/server in DP group} (Thm.~\ref{thm:multi-lb})
& $\bigl(1 + \tfrac{\sum_{i}\frac{\ell_i-1}{\ell_i}}{p} + O\bigl(\tfrac{1}{p^2}\bigr)\bigr) T_0$
& $\bigl(\tfrac{\ell_1}{2} + \tfrac{\ell_1}{2\,p} + O\bigl(\tfrac{1}{p^2}\bigr) \bigr)\,T_0$ \\[6pt]
\textbf{Single straggler, $g$ GPUs/server in DP group} (Thm.~\ref{thm:multigpu-lb})
& $\bigl(1 + \tfrac{g(\ell-1)}{\ell\,p} + O\bigl(\tfrac{g^2}{p^2}\bigr)\bigr) T_0$
& $\bigl(\tfrac{\ell}{2} + \tfrac{\ell}{2\,p} + O\bigl(\tfrac{1}{p^2}\bigr) \bigr)\,T_0$ \\
\bottomrule
\end{tabular}

\vspace{4pt}
\parbox{\columnwidth}{\footnotesize
\textit{Notes.}
(1)~For $m$ stragglers, the two-regime column header indicates the lower bound expression only, not the exact transition point; the precise transition depends on the heterogeneous $\ell_1,\dots,\ell_m$ (see Appendix~\ref{app:proof-multi}).
(2)~For the single-straggler and multi-GPU/server settings, tighter bounds exist; we present these simplified forms for consistency across settings (see Appendix~\ref{app:proof-single-four-v2}, Theorem~\ref{thm:tight-lower-bound} and Appendix~\ref{app:proof-multigpu}, Theorem~\ref{thm:multigpu-tight-lb}). For a complete summary of all theoretical results, see Appendix~\ref{app:optimal}.
}
\end{table}

\section{Background and Motivation}
\label{sec:background}


\begin{figure}[t]
  \centering
  \includegraphics[width=0.6\textwidth]{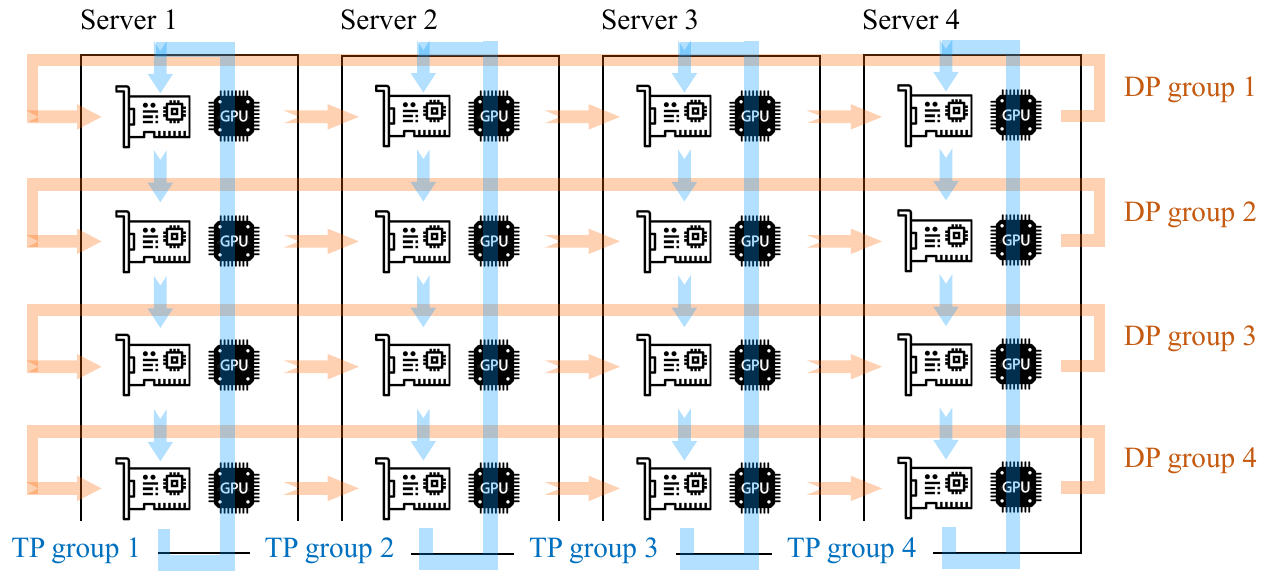}
  \caption{2D parallelism communication topology example for large-model
    training (4 servers $\times$ 4 GPUs). Each
    \textcolor{orange!80!black}{data-parallel (DP) group} consists of
    one GPU per server; gradient synchronization is performed via
    AllReduce over the inter-node RDMA fabric (IB/RoCE). Each
    \textcolor{cyan!60!blue}{tensor-parallel (TP) group} consists of
    the GPUs within a single server, executing collectives over
    intra-node NVLink.}
  \label{fig:allreduce-topology}
\end{figure}

\noindent\textbf{AllReduce topology under hybrid parallelism.}
Modern LLM training combines multiple parallelism
strategies~\cite{narayanan2021megatron,huang2019gpipe,narayanan2019pipedream,zheng2022alpa}---including
tensor parallelism~(TP), data parallelism~(DP), and pipeline
parallelism---of which TP and DP require AllReduce
(Figure~\ref{fig:allreduce-topology}).
Because TP demands an AllReduce at every layer during both the forward
and backward passes, it is confined to GPUs within a single server
to exploit the high-bandwidth NVLink fabric; DP AllReduce, which synchronizes gradients once per
iteration, runs across servers over the inter-node RDMA network
(IB/RoCE).
Each DP group typically occupies only one (or a small number of) GPU
and NIC per server, making DP AllReduce the collective directly
exposed to inter-node network failures.


\begin{figure}[t]
  \centering
  \begin{minipage}[t]{0.55\textwidth}
    \centering
    \includegraphics[width=0.65\linewidth]{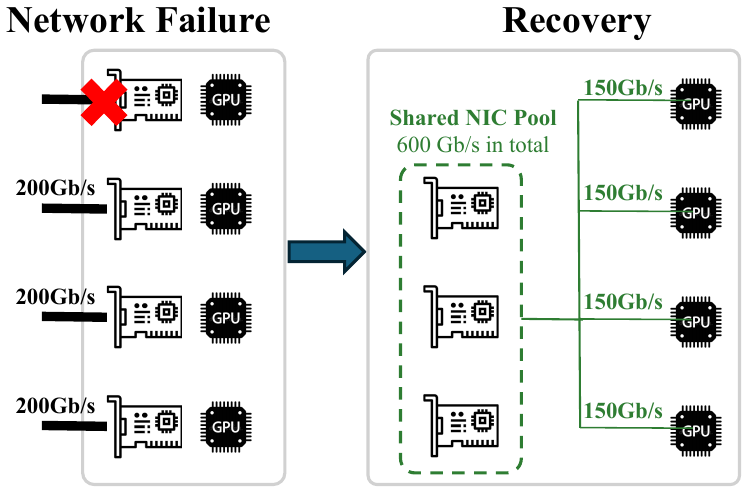}
    \captionof{figure}{Network failure recovery: remaining NICs on
      this server form a \emph{shared NIC pool}, turning the affected
      server into a \emph{straggler} (e.g., losing 2 of 8 NICs gives
      slowdown factor $\ell = 4/3 \approx 1.33$).}
    \label{fig:nic-failure-recovery}
    \vspace{-12pt}
  \end{minipage}\hfill
  \begin{minipage}[t]{0.43\textwidth}
    \centering
    \includegraphics[width=0.765\textwidth]{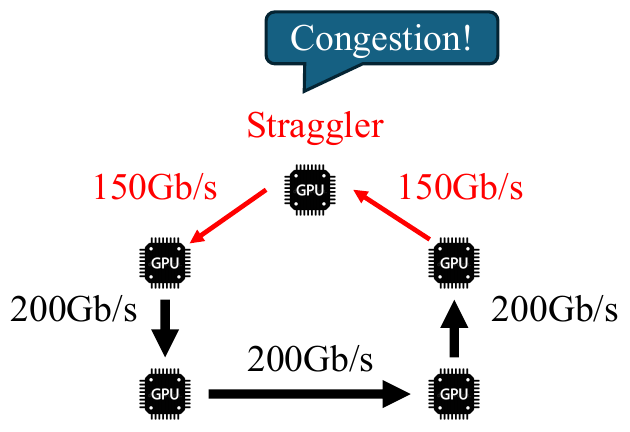}
    \captionof{figure}{The straggler lies on the critical path of
      standard Ring AllReduce, bringing congestion to the entire
      collective.}
    \label{fig:bottleneck}
    \vspace{-12pt}
  \end{minipage}
\end{figure}

\noindent\textbf{Network failure recovery creates stragglers.}
When a network failure (e.g., NIC, cable, or port fault) occurs during
training, there are two strategies to resume without interrupting
training.
The first is to \emph{exclude} the failed server and continue with the
remaining healthy
servers~\cite{ncclx2025, adapcc2024, devraj2025straggler, tandon2017gradient}; this preserves homogeneous bandwidth but
wastes the excluded server's GPUs.
The second is to \emph{keep} the failed server by rerouting its traffic
through surviving network links (i.e., NICs) on the same
server~\cite{r2ccl2025, iccl2025, shift2025}.
As illustrated in Figure~\ref{fig:nic-failure-recovery}, the PXN
(PCI~$\times$~NVLink) mechanism in NCCL~\cite{nccl_pxn2022} allows any
GPU to access any NIC on the same server via NVLink, effectively
pooling all local NICs into a shared bandwidth resource.
This keeps all GPUs active, but leaves the GPUs on the affected server
as stragglers with degraded network bandwidth.


\noindent\textbf{Ring AllReduce bottleneck under asymmetric bandwidth.}
When network bandwidth is the bottleneck---as is the case for large
gradient buffers in data-parallel
training\footnote{We focus on the bandwidth term because (i)~for the
large gradient buffers in LLM training (hundreds of MB to several
GB~\cite{gibiansky2017allreduce}), bandwidth dominates AllReduce
cost---NCCL itself selects the bandwidth-optimal Ring algorithm above
$\sim$256\,KiB~\cite{nvidia2024nccl}---and (ii)~per-hop network
latency varies widely across GPU pairs in a data center and is
difficult to model
analytically~\cite{alibabahpn2024}.}---NCCL
selects the Ring AllReduce
algorithm~\cite{patarasuk2009bandwidth, nvidia2024nccl}.
Ring AllReduce arranges all GPUs in a logical ring and performs a
reduce-scatter followed by an allgather, with each GPU sending and
receiving approximately twice the message size in total.
This is information-theoretically optimal when all links have equal
bandwidth~\cite{patarasuk2009bandwidth}.
However, when one server operates at reduced bandwidth, requiring
every GPU---including the straggler---to send the same volume of data
creates congestion at the straggler's degraded link
(Figure~\ref{fig:bottleneck}).
The straggler becomes the bottleneck of the entire ring, and all
healthy GPUs are throttled to its pace.
Existing liveness-oriented systems do not address this: after failover,
they resume the standard Ring AllReduce without reducing the
straggler's data obligation or restructuring the schedule.
R$^2$CCL's AllReduce optimization~\cite{r2ccl2025}, the most
sophisticated among these, still incurs up to 57\% increase in
AllReduce runtime when 50\% of the bandwidth is lost on a single
straggler.
General-purpose schedule synthesis
tools~\cite{shah2021taccl, cowan2022msccl, wang2020blink, sapio2021switchml} can in
principle handle heterogeneous topologies, but their solver-based
approach requires minutes to hours of computation even at moderate
scale, making them unsuitable for online re-scheduling when a network
failure occurs during training.


\noindent\textbf{Problem settings of AllReduce topology.}
\label{subsec:scenarios}
The topology of the DP AllReduce depends on the parallelism
configuration and the number of concurrent failures; we study three
settings that cover the practical spectrum:
\begin{itemize}[leftmargin=*,itemsep=1pt,topsep=1pt,parsep=0pt,partopsep=0pt]
  \item \textbf{Single straggler, one GPU per server in DP group.}
        In large-scale training of frontier models, all GPUs within a
        server are assigned to TP (and possibly PP), so each server
        contributes exactly one GPU to each DP AllReduce group.
        A single network failure degrades one server, creating one
        straggler in the AllReduce ring.

  \item \textbf{Multiple stragglers, one GPU per server in DP group.}
        In large clusters with hundreds of servers, multiple servers
        may experience network failures and fall back to degraded
        networking~\cite{kokolis2025reliability}.
        This is the most common production scenario; it extends the
        above setting to handle heterogeneous slowdown factors across
        multiple degraded servers.

  \item \textbf{Single straggler, multiple GPUs per server in DP group.}
        For medium-scale training or fine-tuning, TP may use only a
        subset of GPUs within a server (e.g., 2 or 4), leaving multiple
        GPUs per server in the same DP group.
        These GPUs share the server's inter-node NIC bandwidth, so a
        network failure degrades the effective bandwidth for all of them
        simultaneously.
\end{itemize}

\section{Theoretical Lower Bounds}
\label{sec:lower-bound}

We establish information-theoretic lower bounds on AllReduce time under
degraded-NIC conditions, following the general methodology of
communication lower bounds in distributed
learning~\cite{huang2022lower}.
The results stated below correspond to the lower-bound entries in
Table~\ref{tab:overhead}; all proofs
are deferred to Appendix~\ref{app:proof-lb}.

\textbf{Problem setting.}
Let $T(\mathcal{A})$ denote the end-to-end completion time of an
AllReduce algorithm~$\mathcal{A}$ on $p \geq 3$ GPUs with $n$-element
vectors, counted in the bandwidth-bound model: each healthy NIC
transmits one element per time unit, a NIC with slowdown
$\ell > 1$ transmits one element in $\ell$ time units, and per-message
latency, kernel launch, and cold-start overheads are
excluded~\cite{patarasuk2009bandwidth}.
An algorithm is \emph{correct} if (i)~every transfer respects data
dependencies---a GPU may transmit only data it already holds---and
(ii)~upon termination every GPU holds the element-wise sum
$\mathbf{s} = \sum_{i=1}^{p} \mathbf{x}_i$.
An algorithm is \emph{bandwidth-optimal} if it achieves $\min T$ over
all correct AllReduce algorithms on the same bandwidth profile.
In the fault-free homogeneous case with $g$ GPUs per server, this
minimum is known to equal\footnote{When $g > 1$, each server runs $g$
concurrent ring channels, each carrying $n/g$ elements, so that all $g$
NICs per server are fully utilized---this is the standard approach used
by NCCL~\cite{nvidia2024nccl}.}
\[
  T_0 \;=\; \frac{2(p-1)}{gp}\,n\,,
\]
which reduces to $2n(p-1)/p$ when $g = 1$, and is attained by Ring
AllReduce~\cite{patarasuk2009bandwidth}.
We measure the overhead of all algorithms under degraded bandwidth
relative to this~$T_0$.

\paragraph{One straggler.}
Exactly one GPU has slowdown factor~$\ell$, and the remaining
$p - 1$ are healthy.

\begin{theorem}[Lower bound with one straggler]
\label{thm:lower-bound}
Any correct AllReduce algorithm $\mathcal{A}$ satisfies
{\small$$
  T(\mathcal{A}) \;\ge\;
  \max\!\left\{\,\frac{2\ell(p{-}1)}{\ell(p{-}1)+1},\;\ell\,\right\}\cdot n
  \;=\;
  \max\!\left\{\,1 {+} \frac{\ell{-}1}{\ell\,p} {+} O\!\left(\frac{1}{p^2}\right),\;\frac{\ell}{2} {+} \frac{\ell}{2p} {+} O\!\left(\frac{1}{p^2}\right)\,\right\}\cdot T_0.
$$}
\end{theorem}

\paragraph{Multiple stragglers.}
\label{subsec:lb-multi}
More generally, $m$ GPUs are stragglers with heterogeneous slowdown
factors $\ell_1 \geq \ell_2 \geq \cdots \geq \ell_m > 1$;
the remaining $p - m$ are healthy ($p$ is sufficiently larger than $m$).

\begin{theorem}[Lower bound with $m$ stragglers]
\label{thm:multi-lb}
Any correct AllReduce algorithm $\mathcal{A}$ satisfies
{\small\[
  T(\mathcal{A}) \;\geq\;
  \max\!\left\{\,\frac{2(p{-}1)}{p {-} m {+} \sum_{i}\frac{1}{\ell_i}},\;\ell_1\,\right\}\cdot n
  \;=\;
  \max\!\left\{\,1 {+} \frac{\sum_{i}\frac{\ell_i-1}{\ell_i}}{p} {+} O\!\left(\frac{1}{p^2}\right),\;\frac{\ell_1}{2} {+} \frac{\ell_1}{2p} {+} O\!\left(\frac{1}{p^2}\right)\,\right\}\cdot T_0.
\]}
Setting $m = 1$ recovers Theorem~\ref{thm:lower-bound}.
\end{theorem}

\paragraph{Multiple GPUs per server.}
\label{subsec:lb-multigpu}
Now we consider each server holding $g$ GPUs that share the same NIC
bandwidth. Inside each server, the $g$ GPUs are connected via NVLink whose
bandwidth is at least $(g{-}1)$ times the NIC
bandwidth.\footnote{On the NVIDIA DGX A100~\cite{nvidia_dgx_a100},
each GPU has 4800\,Gbps aggregate NVLink bandwidth
(2400\,Gbps unidirectional) versus 200\,Gbps inter-node NIC bandwidth
(one HDR InfiniBand NIC per GPU),
so this assumption holds by a wide margin.}
The straggler server's inter-server NIC is degraded by factor~$\ell$, while
all NVLink connections remain healthy.

\begin{theorem}[Lower bound with multi-GPU servers]
\label{thm:multigpu-lb}
Under the NVLink-rich assumption above, any correct AllReduce algorithm
$\mathcal{A}$ satisfies
{\footnotesize$$
  T(\mathcal{A}) \;\geq\;
  \max\!\left\{\,\frac{2\ell(p-g)}{g\,(\ell(p-g)+g)},\;\;\frac{\ell}{g}\,\right\}\cdot n
  \;=\;
  \max\!\left\{\,1 + \frac{g(\ell-1)}{\ell\,p} + O\!\left(\frac{g^2}{p^2}\right),\;\;\frac{\ell}{2} + \frac{\ell\,g}{2\,p} + O\!\left(\frac{g^2}{p^2}\right)\,\right\}\cdot T_0\,.
$$}
Setting $g = 1$ recovers Theorem~\ref{thm:lower-bound}.
\end{theorem}

\paragraph{Takeaway.}
In all three settings, the lower bound has the form
$T \geq \max\{1 + O(1/p),\;\ell/2 + O(1/p)\}\cdot T_0$.
When the worst straggler loses more than half of its bandwidth,
it bottlenecks the entire AllReduce and any algorithm requires
at least $(\ell/2)\,T_0$ time.
When no straggler loses more than half of its bandwidth ($\ell < 2$),
the lower bound approaches~$T_0$ as $p$ grows, meaning the overhead
from NIC degradation vanishes in large clusters.

\section{Algorithm Design}
\label{sec:algorithm}

Our algorithm partitions each GPU's $n$-element vector into many small
segments and, for every segment, decomposes the AllReduce into four
logical stages: (1)~reduce-scatter among the $p{-}1$ healthy GPUs,
(2)~upload partial sums to the straggler, (3)~download global sums from
the straggler, and (4)~allgather among the healthy GPUs. Stages~2
and~3 use only the slow straggler link while Stages~1 and~4 use fast links between healthy
GPUs---these are \emph{disjoint hardware resources}. We
exploit this by pipelining segments so that the slow Stages~2/3 of one
segment run in parallel with Stages~1/4 of other segments, effectively
hiding the straggler overhead.
We present the algorithm for the single-straggler, one-GPU-per-server
setting in this section; the algorithm designs for multiple stragglers
(Appendix~\ref{app:multi-straggler-algo}) and multi-GPU-per-server
(Appendix~\ref{app:multi-gpu}) are deferred to the appendix.

\subsection{Four-Stage Decomposition}
\label{subsec:four-stage}

We partition the $n$ input elements into $k$ \emph{segments}, each
containing $p-1$ \emph{sections} of size $s = n / (k(p-1))$ elements.
For each segment, the AllReduce proceeds through four stages
(illustrated in Figure~\ref{fig:four-stages}).
We define a \emph{flow} as a point-to-point transfer of one section between
two GPUs, and enforce that each NIC receives at most one flow at a time;
overlapping incoming flows would contend for bandwidth and cause unpredictable
delays.

\definecolor{healthygreen}{HTML}{B4E5A1}
\definecolor{stragglerpeach}{HTML}{F6C6AD}
\begin{figure*}[t]
  \centering
  \includegraphics[width=0.24\textwidth,page=1,trim=371 208 422 175,clip]{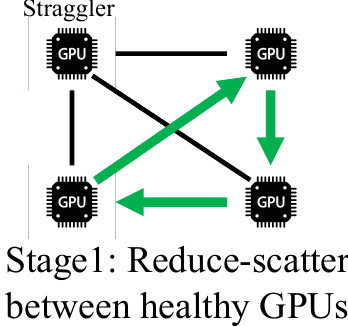}\hfill
  \includegraphics[width=0.24\textwidth,page=2,trim=371 208 422 175,clip]{figures/three_stage_pipeline.pdf}\hfill
  \includegraphics[width=0.24\textwidth,page=3,trim=371 208 422 175,clip]{figures/three_stage_pipeline.pdf}\hfill
  \includegraphics[width=0.24\textwidth,page=4,trim=371 208 422 175,clip]{figures/three_stage_pipeline.pdf}
  \caption{4-stage decomposition of \textsc{OptCC}. Line thickness
    indicates bandwidth;
    \textcolor{green!50!black}{\textit{arrows}} denote
    \textcolor{green!50!black}{\textit{active}} links in each stage.}
  \label{fig:four-stages}
\end{figure*}

\paragraph{Stage~1: Reduce-Scatter among healthy servers.}
The $p-1$ healthy servers are arranged in a directed ring. Each section is
reduce-scattered along this ring in exactly $p-2$ hops: at each hop, a server
receives a partial-sum flow from its predecessor, adds its own contribution, and
forwards the result. After $p-2$ hops, one designated healthy server holds the
partial sum of all $p-1$ healthy servers' contributions for that section.

Since all $p-1$ servers send in parallel (each to its ring successor), and each
hop transmits one section of $s$~elements on a fast link, the time for Stage~1
is $(p-2) \cdot s$.

\paragraph{Stage~2: Upload to the straggler.}
After Stage~1, the partial sum of a section resides on a healthy server but does
not yet include the straggler's data. In Stage~2, a healthy server sends the
partial sum to the straggler ($s$~elements on the slow link, taking
$\ell \cdot s$ time).

\paragraph{Stage~3: Download from the straggler.}
The straggler adds its local value to the received partial sum, producing the
global sum, and sends it back to a healthy server ($s$~elements, again
$\ell \cdot s$ time).

Each segment has $p-1$ sections, and the flow scheduling constraint forces the
straggler to process them one at a time: no two flows may overlap at the
straggler's NIC. The total time for Stages~2 and~3 for one segment is therefore
$(p-1) \cdot \ell \cdot s$ each. Because the straggler can receive and send in
full-duplex, Stage~2 and Stage~3 for \emph{different} sections can be pipelined
on the straggler's link.

\paragraph{Stage~4: Allgather among healthy servers.}
Once a section's global sum returns to a healthy server via Stage~3, it is
broadcast to all other healthy servers through an allgather along the same ring.
This is structurally symmetric to Stage~1: $p-2$ hops, each forwarding
$s$~elements, costing $(p-2) \cdot s$.

\paragraph{Stage ordering and correctness.}\label{par:stage-ordering}
For a single section, the four stages must respect data dependencies: Stage~1
produces the partial sum consumed by Stage~2, Stage~2 feeds Stage~3, and
Stage~3 produces the global sum consumed by Stage~4. Two orderings satisfy
these constraints:
(1)~S1${}\to{}$S2${}\to{}$S3${}\to{}$S4, and
(2)~S3${}\to{}$S1${}\to{}$S4${}\to{}$S2.
In ordering~(1), the healthy servers first reduce-scatter among
themselves, then upload the partial sum to the straggler, which completes the
global sum and downloads it for allgather. In ordering~(2), the straggler
uploads its local value \emph{first}; the healthy servers incorporate it during
their reduce-scatter (which now produces the global sum directly), perform the
allgather, and finally send the result back to the straggler. Both orderings compute a correct AllReduce.

\subsection{Schedule Construction}
\label{subsec:schedule}

A key observation is that Stages~2 and~3 communicate exclusively over the slow
straggler link, whereas Stages~1 and~4 communicate exclusively over the
fast ring links among healthy servers. Although individual GPUs may serve
as endpoints in both types of transfer, the \emph{links} involved are
disjoint. This separation of resources makes it possible to overlap
Stages~2/3 with Stages~1/4: if we execute them concurrently, each
parallel body takes time proportional to $\max\{\ell,\, 2\} \cdot s$,
which aligns precisely with the lower bound established in
Theorem~\ref{thm:lower-bound}.

\paragraph{Four patterns.}
Exploiting the two valid stage orderings identified in
Section~\ref{par:stage-ordering}, we design four \emph{patterns}---distinct
flow schedules that can be interleaved without violating the non-overlap
constraint. We illustrate the construction for $p = 5$ and $\ell = 2$
(Figure~\ref{fig:four-patterns}); the analysis generalizes to arbitrary
$p$ and $\ell$.

In this setting, GPU~0 is the straggler with 50\% NIC throughput; GPUs~1--4
are healthy. Each Stage~2 or Stage~3 flow occupies $\ell\,s = 2s$ time units
(two grid cells in the figure), while each Stage~1 or Stage~4 flow occupies
$s$ time units (one grid cell). With $p - 1 = 4$ healthy servers, Stage~1
(reduce-scatter among 4~GPUs) and Stage~4 (allgather among 4~GPUs) each
consist of $(p-2)(p-1) = 12$ flows; Stage~2 and Stage~3 each consist of
$p - 1 = 4$ flows. We reserve $2(p{-}1) = 8$ time slots per stage (one
grid cell equals the time for a healthy NIC to send or receive one flow;
a straggler flow occupies two cells), for a total of
$4 \times 2(p{-}1) = 32$ time slots per pattern across all four stages.
The four patterns are:
\begin{itemize}[leftmargin = *]
  \item \textbf{Pattern~A} (Figure~\ref{fig:pattern-a}): S1 $\to$ S2 $\to$ S3 $\to$ S4.
  \item \textbf{Pattern~B} (Figure~\ref{fig:pattern-b}): S3 $\to$ S1 $\to$ S4 $\to$ S2.
  \item \textbf{Pattern~C} (Figure~\ref{fig:pattern-c}): S1 $\to$ S2 $\to$ S3 $\to$ S4
        (same stage ordering as~A, but flows in Stages~1 and~3
        are shifted by an offset).
  \item \textbf{Pattern~D} (Figure~\ref{fig:pattern-d}): S3 $\to$ S1 $\to$ S4 $\to$ S2
        (same ordering as~B, with the same offset).
\end{itemize}
Each pattern respects the data dependencies of its stage ordering, and the
four patterns are staggered so that all four stages interleave across
time---no NIC ever receives or sends two simultaneous flows.

\begin{figure*}[t]
  \centering
  \begin{subfigure}[b]{0.48\textwidth}
    \centering
    \includegraphics[width=\textwidth]{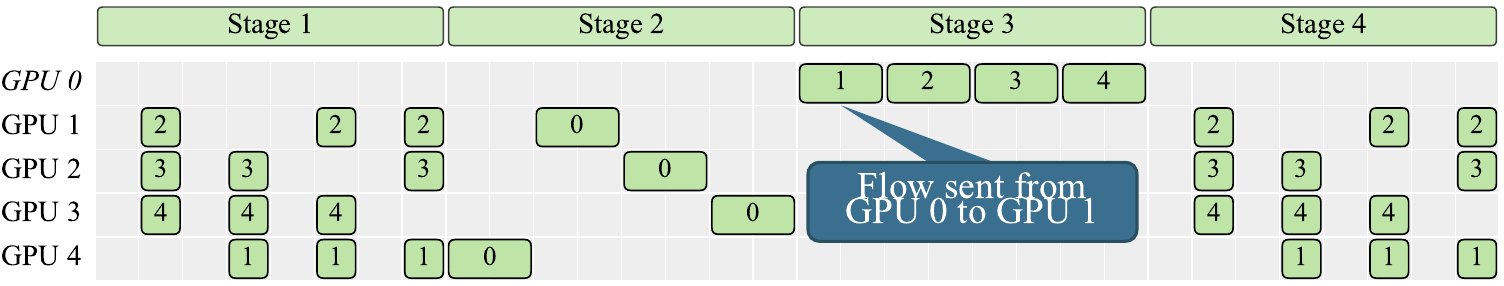}
    \caption{Pattern~A (S1$\to$S2$\to$S3$\to$S4)}
    \label{fig:pattern-a}
  \end{subfigure}\hfill
  \begin{subfigure}[b]{0.48\textwidth}
    \centering
    \includegraphics[width=\textwidth]{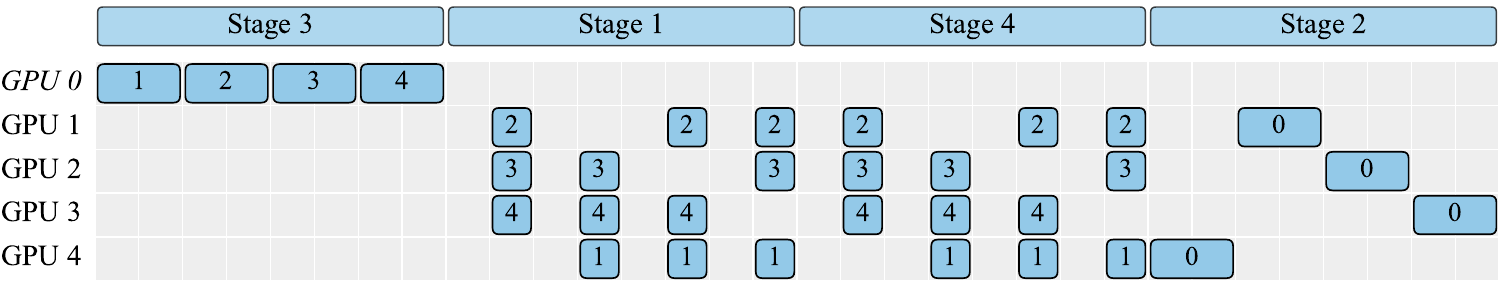}
    \caption{Pattern~B (S3$\to$S1$\to$S4$\to$S2)}
    \label{fig:pattern-b}
  \end{subfigure}\\[6pt]
  \begin{subfigure}[b]{0.48\textwidth}
    \centering
    \includegraphics[width=\textwidth]{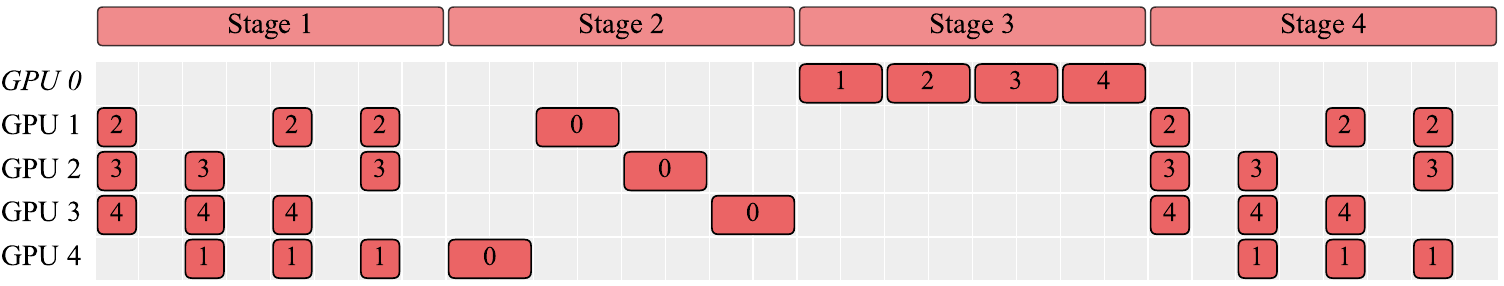}
    \caption{Pattern~C (S1$\to$S2$\to$S3$\to$S4)}
    \label{fig:pattern-c}
  \end{subfigure}\hfill
  \begin{subfigure}[b]{0.48\textwidth}
    \centering
    \includegraphics[width=\textwidth]{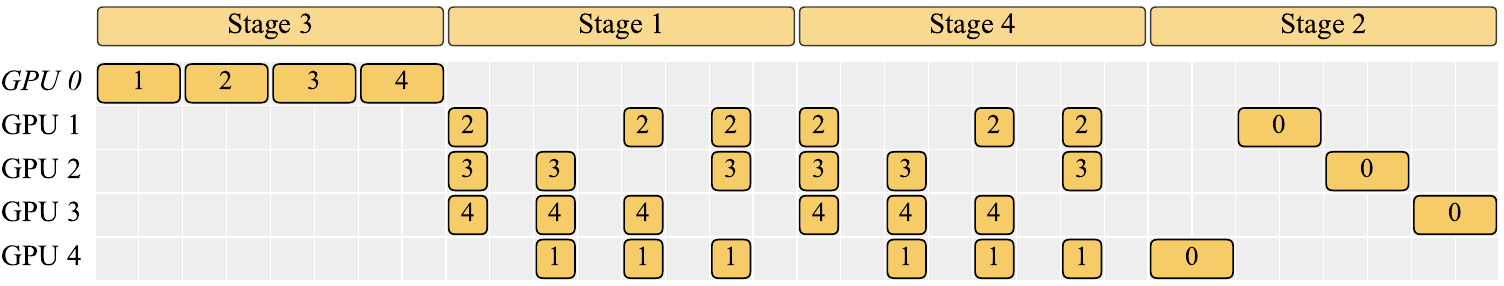}
    \caption{Pattern~D (S3$\to$S1$\to$S4$\to$S2)}
    \label{fig:pattern-d}
  \end{subfigure}
  \caption{Flow schedules for the four patterns ($p{=}5$, $\ell{=}2$).
    Rows are GPUs~0--4 (\textit{GPU~0} is the straggler, losing 50\% bandwidth); columns are time slots.
    Each colored cell represents a flow; its label indicates the destination
    GPU. The four patterns occupy disjoint communication slots so that no NIC
    receives two flows simultaneously.}
  \label{fig:four-patterns}
\end{figure*}

\paragraph{Pipelining four patterns.}
We now overlay the four patterns to form a complete schedule
(Figure~\ref{fig:patterns-combined}). We assume $k$ is a multiple of~4, so
that each pattern processes exactly $k/4$ segments. Superimposing the four
patterns' flow schedules yields a composite timeline in which each
parallel body contains one instance of each stage from each pattern.

\begin{figure*}[t]
  \centering
  \includegraphics[width=\textwidth]{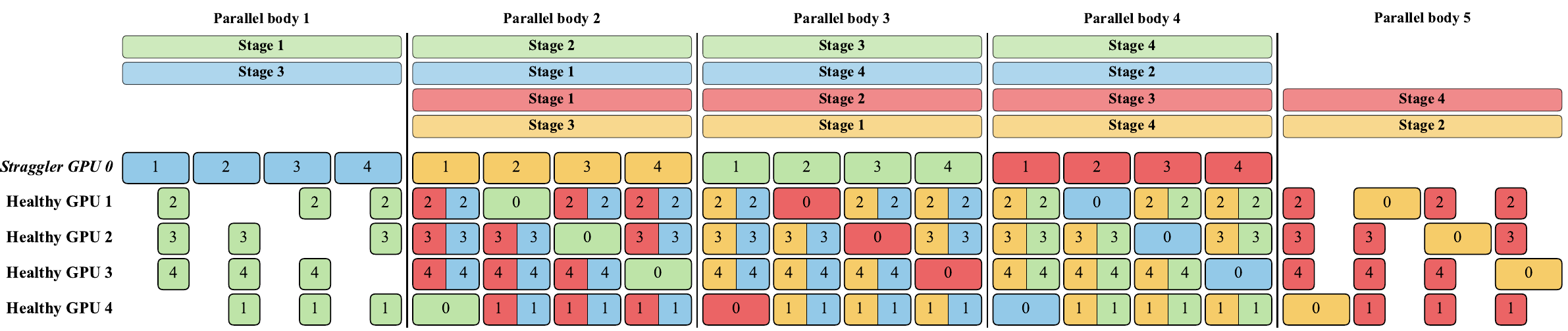}
  \caption{\textbf{Combining the four patterns from
    Figure~\ref{fig:four-patterns} into a complete pipeline.}
    $p{=}5$ GPUs, $\ell{=}2$ (straggler lost 50\% bandwidth), $k{=}4$ segments.
    Each parallel body contains all four stages
    (except the head and tail).
    Flows enclosed in a single rounded box
    share the same sender--receiver pair and have no data dependency,
    thus can be packed into one network transfer.
    \textbf{When $\ell{=}2$, the body has no idle bubbles.}}
  \label{fig:patterns-combined}
\end{figure*}

\paragraph{Filling straggler-slot bubbles ($\ell < 2$).}
When $\ell < 2$, each Stage~2/3 flow ($\ell\,s$) is shorter than the
combined Stage~1/4 flows ($2s$), leaving a bubble of $(2-\ell)\,s$ time
in every straggler slot. Across one parallel body the straggler has
$p{-}1$ such bubbles, and each healthy GPU also has one (in its Stage~2/3
slot). We reclaim these idle intervals by running a \emph{point-to-point
allreduce} between the healthy GPUs and the straggler
(Figure~\ref{fig:pipeline-l-lt-2}): in parallel body~$b$, each healthy
GPU uses its bubble to send a partial sum to the straggler
({\definecolor{lgray}{gray}{0.85}\textcolor{lgray}{\rule{0.8em}{0.8em}}\,light gray} cells); the straggler reduces these
contributions during its own bubbles and, in body~$b{+}1$, broadcasts
the result back to the healthy GPUs
({\definecolor{dgray}{gray}{0.5}\textcolor{dgray}{\rule{0.8em}{0.8em}}\,dark gray} cells). Because both directions reuse
time that would otherwise be idle, a fraction of the data completes its
allreduce at zero additional cost, bringing the effective runtime closer
to the lower bound.

\definecolor{lgray}{gray}{0.85}
\definecolor{dgray}{gray}{0.5}
\begin{figure*}[t]
  \centering
  \includegraphics[width=\textwidth]{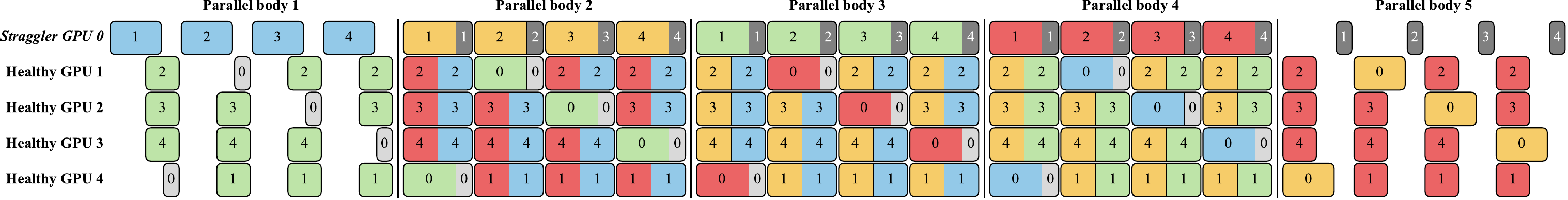}
  \caption{\textbf{Filling all bubbles with P2P allreduce between healthy
    GPUs and the straggler when $\ell < 2$.}
    $p{=}5$ GPUs, $\ell{=}1.5$ (straggler lost 33\% bandwidth), $k{=}4$ segments.
    Stage~2/3 flows (width~$1.5$) are shorter
    than Stage~1/4 flows (width~$2$), creating bubbles.
    Filling bubbles with:
    \textcolor{lgray}{\rule{0.8em}{0.8em}}\,light gray (in parallel body $i$): healthy GPUs send partial sums to the
    straggler;
    \textcolor{dgray}{\rule{0.8em}{0.8em}}\,dark gray (in parallel body $i{+}1$): straggler broadcasts reduced results
    back to the healthy GPUs.}
  \label{fig:pipeline-l-lt-2}
\end{figure*}

\subsection{Time Analysis}
\label{subsec:time-analysis}

Let $s = n / (k(p-1))$ denote the section size. The composite schedule
is formed by concatenating $k/4$ copies of the four-pattern group from
Figure~\ref{fig:patterns-combined} (for $\ell\ge 2$) or
Figure~\ref{fig:pipeline-l-lt-2} (for $\ell<2$) end-to-end, where the
last parallel body of one group overlaps with the first of the next.
This produces $k+1$ parallel bodies in total: $k-1$ steady-state bodies
that each contain one part from each of the four patterns, plus one
startup body at the head and one drain body at the tail. The body
duration depends on which resource---the healthy ring (Stages~1/4) or
the straggler link (Stages~2/3)---is the bottleneck.

\paragraph{Case $\ell \ge 2$ (straggler-bottlenecked).}
When $\ell \ge 2$, Stages~2/3 are the critical path. Each parallel body
takes $\ell\,(p{-}1)\,s$ time. Summing over $k+1$ bodies (including
startup/drain), substituting $s = n/(k(p-1))$, and recalling the
fault-free baseline $T_0 = 2n(p{-}1)/p$:
\begin{equation}
\label{eq:total-time-large-l}
  T \;=\; \ell\,n\,\frac{k+1}{k}
    \;\xrightarrow{k \to \infty}\; \ell\,n
    \;=\; \left(\frac{\ell}{2} + \frac{\ell}{2\,p} + O\!\left(\frac{1}{p^2}\right)\right)\,T_0\,,
\end{equation}
matching the lower bound of Theorem~\ref{thm:lower-bound}
(Table~\ref{tab:overhead}).
When $\ell = 2$, Stages~2/3 fit exactly
within Stages~1/4 and every slot is occupied with zero idle time. When $\ell > 2$, the straggler becomes the bottleneck, yet the \emph{schedule remains bandwidth-optimal without requiring $p \to \infty$}, confirming its optimality even for small clusters.

\paragraph{Case $\ell < 2$ (ring-bottlenecked, with bubble filling).}
When $\ell < 2$, the bubble-filling technique of
Section~\ref{subsec:schedule} reduces the section size from~$s$ to a
smaller~$s'$, and uses up all the straggler's bandwidth.
The derivation of~$s'$ and the resulting total time are given in
Appendix~\ref{app:bubble-analysis}:
\begin{equation}
\label{eq:total-time-small-l}
  T = \frac{2(p{-}1)\,\ell}{(p{-}2)\ell+2}\,n
       \,\frac{k+\ell-1}{k} \;\xrightarrow{k\to\infty}\;
     \frac{2(p{-}1)\,\ell\,n}{(p{-}2)\ell+2}
     \;=\; \left(1 + \frac{2(\ell{-}1)}{\ell\,p}
           + O\!\left(\frac{1}{p^2}\right)\right)\,T_0\,.
\end{equation}
At $\ell=1$ this reduces to $T_0$, the non-degraded optimum;
at $\ell=2$ it gives $2n$, approaching the lower bound of
Theorem~\ref{thm:lower-bound}.
For general $\ell\in[1,2)$, the relative overhead satisfies
$p\ell/[(p{-}2)\ell+2] \le 1+1/(p{-}1)$.
In fact, our algorithm always \emph{achieves bandwidth optimal without requiring $p \to \infty$}; see Appendix~\ref{app:bubble-analysis} for details.

\paragraph{Schedule generation complexity.}
The complete flow schedule is determined by the closed-form rules of
Sections~\ref{subsec:four-stage}--\ref{subsec:schedule} in $O(pk)$
time, with no iterative optimization.
For $p{=}1024$ GPUs on SimAI, the schedule is generated in under
1\,ms on a single CPU core, making OptCC suitable for online
deployment upon failure detection.

\subsection{Extensions to Multiple Stragglers and Multi-GPU Servers}
\label{subsec:extensions}

\paragraph{Multiple stragglers, one GPU per server in DP group (Appendix~\ref{app:multi-straggler-algo}).}
\label{subsec:multi-straggler}
The key difference is that Stages~2/3 now serve $m$
stragglers instead of one; the bottleneck is the slowest straggler
($\ell_1 = \max_i \ell_i$), and flows to/from each straggler are
assigned to disjoint time slots within each body.

\paragraph{Single straggler, multiple GPUs per server in DP group (Appendix~\ref{app:multi-gpu}).}
\label{subsec:multi-gpu}
We treat each server's $g$~GPUs as a single \emph{logical GPU}: before
every outgoing NIC flow the $g$~GPUs perform a local reduce over
NVLink, and after every incoming flow the receiver broadcasts to the
other $g{-}1$ GPUs.
Because NVLink bandwidth is typically $(g{-}1){\times}$ the NIC
bandwidth, these intra-server transfers can be fully \emph{overlapped}
with inter-server NIC transfers of adjacent pipeline stages, so the
NIC remains the sole bottleneck and intra-server communication adds no
extra time to the critical path.

\section{Evaluation}
\label{sec:experiments}

We evaluate OptCC on SimAI~\cite{simai2025} (an NS-3-based network
simulator that models the actual on-the-wire traffic between GPUs),
using clusters of A100 servers, each equipped with 8~GPUs and
8~ConnectX-6 200\,Gbps HDR InfiniBand NICs.
We compare against three baselines:
\textbf{NCCL\textsubscript{NoFailure}}, NCCL ring on a
fully healthy topology (a fault-free reference);
\textbf{ICCL}~\cite{iccl2025}, NCCL ring on the NIC-bandwidth-degraded
topology; and
\textbf{R$^2$CCL}~\cite{r2ccl2025}, a state-of-the-art NIC
fault-tolerant AllReduce.
We also plot the \textbf{lower bounds} from
Section~\ref{sec:lower-bound} as theoretical runtime
guarantees.\footnote{SimAI profiling shows that the
NIC bandwidth term accounts for $\sim$86\% of
NCCL\textsubscript{NoFailure}'s completion time, with the remaining
$\sim$14\% from fixed overhead (cold start, per-flow network latency).
We plot the lower bound as its multiplier applied to the bandwidth
fraction, plus the fixed overhead.}
In the OptCC implementation, flow dependencies are enforced via
per-GPU state tracking for intra-GPU ordering and necessary P2P
synchronization messages for inter-GPU ordering.
We simulate at production-realistic GPU
counts~\cite{llama3herd2024,narayanan2021megatron} and message
sizes~\cite{pytorch_ddp_bucket2024, vogels2019powersgd, lin2018deep}.
SimAI is deterministic---the same configuration always produces the
same result---so each data point is a single run with no error bars.
Larger-message microbenchmarks are in Appendix~\ref{app:large-n-eval}.

\begin{figure*}[!ht]
  \centering
  \includegraphics[width=0.75\textwidth]{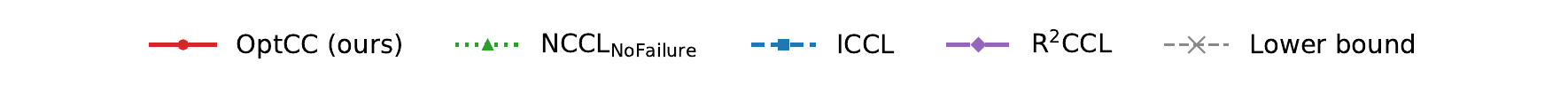}\\[-6pt]
  \begin{subfigure}[b]{0.19\textwidth}
    \centering
    \includegraphics[width=\textwidth]{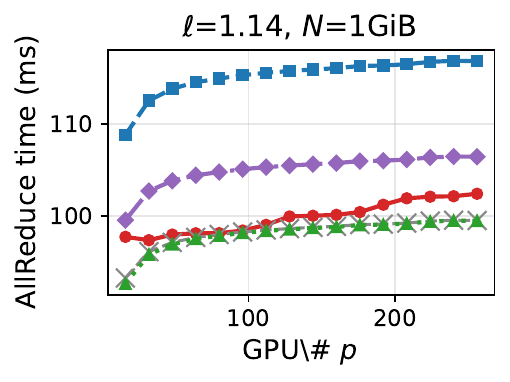}
    \vspace{-20pt}\caption{}
  \end{subfigure}\hfill
  \begin{subfigure}[b]{0.19\textwidth}
    \centering
    \includegraphics[width=\textwidth]{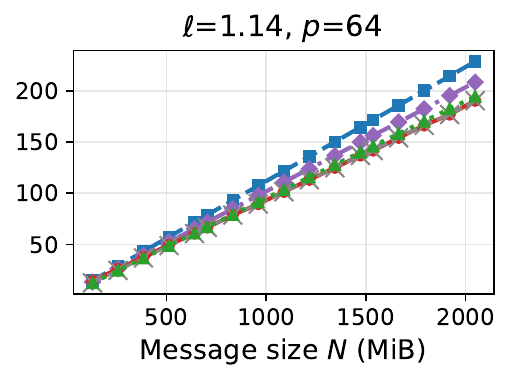}
    \vspace{-20pt}\caption{}
  \end{subfigure}\hfill
  \begin{subfigure}[b]{0.19\textwidth}
    \centering
    \includegraphics[width=\textwidth]{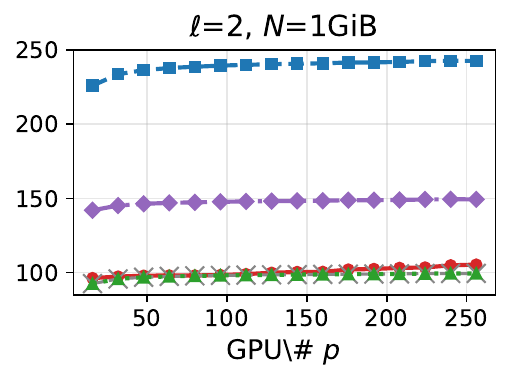}
    \vspace{-20pt}\caption{}
  \end{subfigure}\hfill
  \begin{subfigure}[b]{0.19\textwidth}
    \centering
    \includegraphics[width=\textwidth]{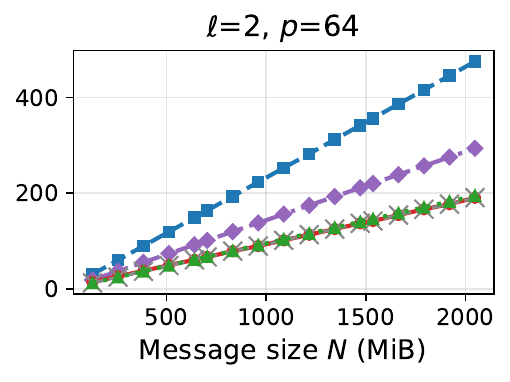}
    \vspace{-20pt}\caption{}
  \end{subfigure}\hfill
  \begin{subfigure}[b]{0.19\textwidth}
    \centering
    \includegraphics[width=\textwidth]{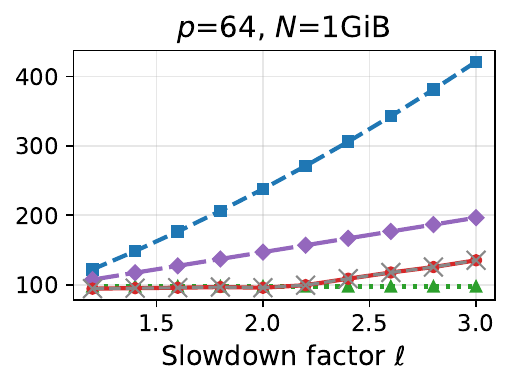}
    \vspace{-20pt}\caption{}
  \end{subfigure}
  \vspace{-6pt}
  \caption{Single straggler; DP group involves 1 GPU per server.
    (a,b)~the straggler loses 1 out of 8 NICs ($\ell{=}1.14$);
    (c,d)~loses 4 out of 8 NICs ($\ell{=}2$);
    (e)~varying~$\ell$.}
  \label{fig:eval-single}
\end{figure*}

\begin{figure*}[!ht]
  \vspace{-8pt}
  \centering
  \begin{subfigure}[b]{0.19\textwidth}
    \centering
    \includegraphics[width=\textwidth]{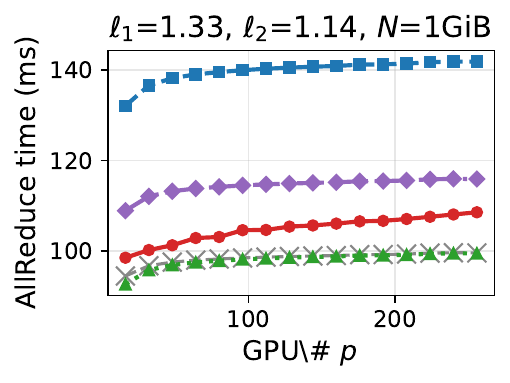}
    \vspace{-20pt}\caption{}
  \end{subfigure}\hfill
  \begin{subfigure}[b]{0.19\textwidth}
    \centering
    \includegraphics[width=\textwidth]{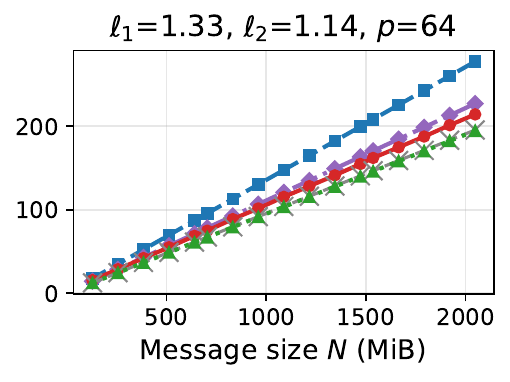}
    \vspace{-20pt}\caption{}
  \end{subfigure}\hfill
  \begin{subfigure}[b]{0.19\textwidth}
    \centering
    \includegraphics[width=\textwidth]{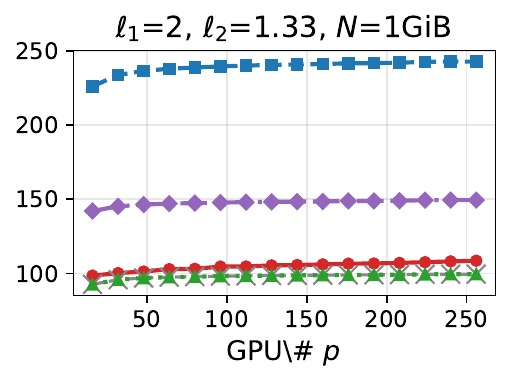}
    \vspace{-20pt}\caption{}
  \end{subfigure}\hfill
  \begin{subfigure}[b]{0.19\textwidth}
    \centering
    \includegraphics[width=\textwidth]{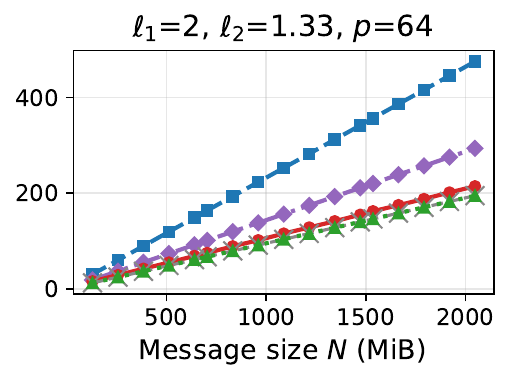}
    \vspace{-20pt}\caption{}
  \end{subfigure}\hfill
  \begin{subfigure}[b]{0.19\textwidth}
    \centering
    \includegraphics[width=\textwidth]{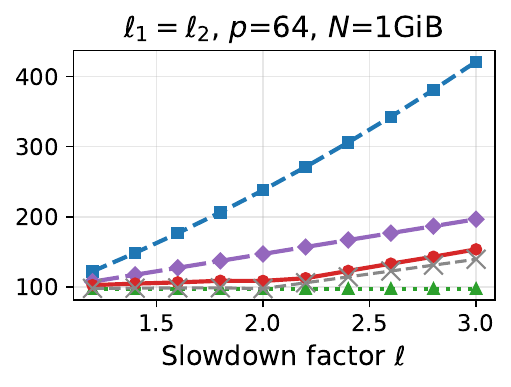}
    \vspace{-20pt}\caption{}
  \end{subfigure}
  \vspace{-6pt}
  \caption{Multi-straggler ($m{=}2$); DP group involves 1 GPU per server.
    (a,b)~two stragglers lose 2 and 1 out of 8 NICs respectively
    ($\ell_1{=}1.33$, $\ell_2{=}1.14$);
    (c,d)~lose 4 and 2 out of 8 NICs
    ($\ell_1{=}2$, $\ell_2{=}1.33$);
    (e)~$\ell_1{=}\ell_2$, varying~$\ell$.}
  \label{fig:eval-multi}
\end{figure*}

\begin{figure*}[!ht]
  \vspace{-8pt}
  \centering
  \begin{subfigure}[b]{0.19\textwidth}
    \centering
    \includegraphics[width=\textwidth]{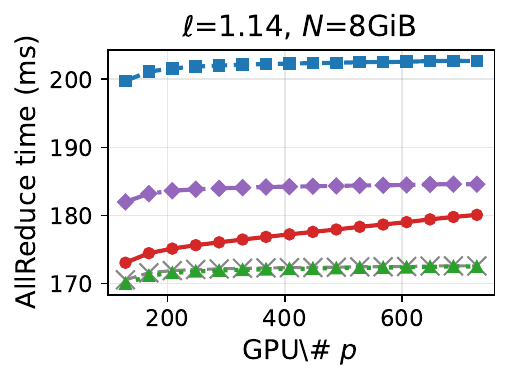}
    \vspace{-20pt}\caption{}
  \end{subfigure}\hfill
  \begin{subfigure}[b]{0.19\textwidth}
    \centering
    \includegraphics[width=\textwidth]{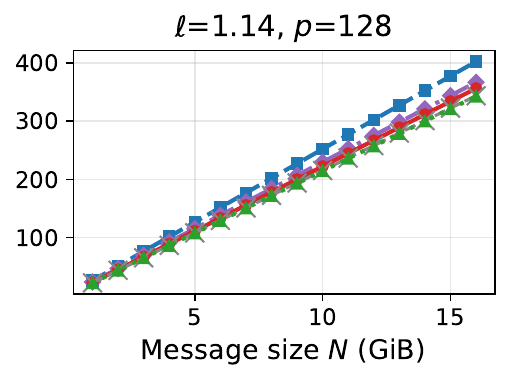}
    \vspace{-20pt}\caption{}
  \end{subfigure}\hfill
  \begin{subfigure}[b]{0.19\textwidth}
    \centering
    \includegraphics[width=\textwidth]{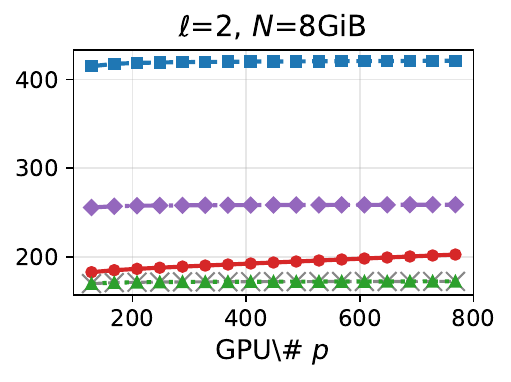}
    \vspace{-20pt}\caption{}
  \end{subfigure}\hfill
  \begin{subfigure}[b]{0.19\textwidth}
    \centering
    \includegraphics[width=\textwidth]{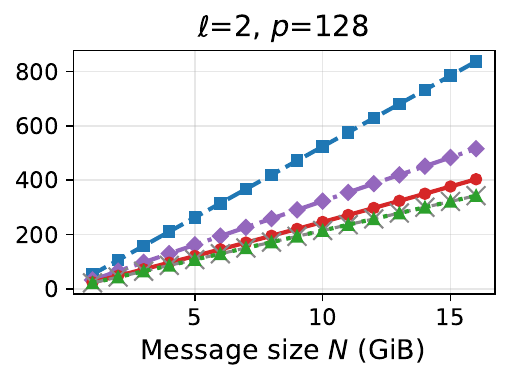}
    \vspace{-20pt}\caption{}
  \end{subfigure}\hfill
  \begin{subfigure}[b]{0.19\textwidth}
    \centering
    \includegraphics[width=\textwidth]{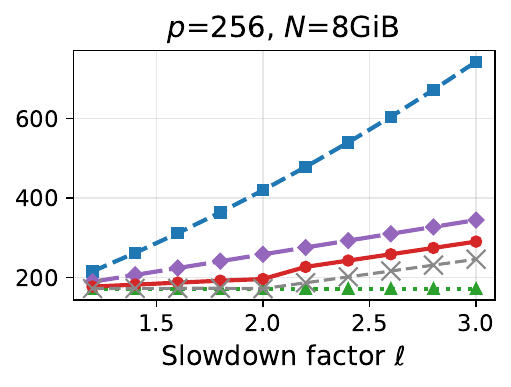}
    \vspace{-20pt}\caption{}
  \end{subfigure}
  \vspace{-6pt}
  \caption{Multi-GPU/server; DP group involves $g{=}4$ GPUs per server.
    (a,b)~the straggler loses 1 out of 8 NICs ($\ell{=}1.14$);
    (c,d)~loses 4 out of 8 NICs ($\ell{=}2$);
    (e)~varying~$\ell$.}
  \label{fig:eval-mgpu}
\end{figure*}

\noindent\textbf{OptCC runtime is close to NCCL\textsubscript{NoFailure} in
both small and large clusters.}
OptCC stays within 6\% of
NCCL\textsubscript{NoFailure} across $p\in[16,256]$. At small~$p$ (e.g.,
$p{=}16$) the gap is dominated by the straggler-induced lower bound
of Theorem~\ref{thm:lower-bound} --- fundamental, not an artifact of
OptCC. Toward $p{=}256$ we observe a mild upward trend; we honestly
acknowledge this as a limitation of our construction, since larger
clusters require more point-to-point synchronization between healthy
and straggler GPUs, whose overhead grows with~$p$. Even so, at
$p{=}256$ OptCC stays within 6\% of
NCCL\textsubscript{NoFailure} and beats R$^2$CCL/ICCL by
29\%/58\%.
Note that ICCL's overhead exceeds the slowdown factor~$\ell$ itself:
because ICCL does not restructure the ring schedule, the straggler's
degraded NIC must carry the same volume of traffic as a healthy NIC,
causing network congestion at the straggler that amplifies the
bandwidth penalty beyond the raw $\ell{\times}$ slowdown.

\noindent\textbf{OptCC approaches NCCL\textsubscript{NoFailure} across all
AllReduce message sizes.}
As long as the worst NIC retains at least 50\% of its bandwidth, OptCC
matches the NCCL\textsubscript{NoFailure} rate to within
2\% (single straggler), 8.5\%
(multi-straggler $m{=}2$), and 4\%
(multi-GPU/server), in each case beating R$^2$CCL by
27\% and ICCL by 55\% at the
same~$N$. The constant-factor gap translates linearly into wall-clock
savings as gradient buffers grow.

\noindent\textbf{OptCC approaches NCCL\textsubscript{NoFailure} whenever
the worst NIC loses at most 50\% of its bandwidth.}
For $\ell\le 2$ the OptCC curve sits flat against
NCCL\textsubscript{NoFailure}: bubble filling
(Section~\ref{subsec:schedule}) hides the entire straggler overhead
inside slack on the healthy ring. Once $\ell$ crosses~$2$ the
straggler's NIC becomes the binding bottleneck and OptCC grows
linearly in~$\ell$ at exactly the rate of Theorem~\ref{thm:lower-bound}'s
$\ell\,n$ branch; even here OptCC stays strictly faster than ICCL and
R$^2$CCL, both of which pay a slowdown penalty from $\ell{=}1$ onward.



\section{Conclusion}
\label{sec:conclusion}

Network failures are an operational reality in large GPU clusters, yet
their impact on collective communication has lacked a principled
theoretical treatment.
This paper closes that gap for AllReduce: we derive tight
information-theoretic lower bounds showing that, when the straggler
retains at least half of its bandwidth, the unavoidable overhead is
only $O(1/p)$---fundamentally negligible in large clusters---and we
design OptCC, an algorithm that achieves these bounds.
The key insight is that healthy GPUs carry enough spare bandwidth to
absorb a straggler's deficit through careful pipeline interleaving,
so that near-optimal performance is attainable \emph{without any
hardware redundancy, job restart, or offline pre-computation}.
SimAI-based evaluation confirms 2--6\% overhead relative to the
fault-free NCCL ring, compared with up to 57\% for existing schemes.
Extending the bounds to latency-sensitive settings and dynamic
bandwidth fluctuations is left to future work.

\bibliographystyle{plainnat}
\bibliography{references}

\appendix

\section{Summary of Optimality Results}
\label{app:optimal}

This appendix collects the theoretical lower bounds, algorithm time complexities, and optimality results presented throughout the paper. Table~\ref{tab:lower-bounds-summary} summarizes the lower bounds, Table~\ref{tab:alg-time-summary} summarizes the algorithm time complexities, and Table~\ref{tab:optimality-summary} summarizes the optimality results.

Recall the problem setting. Let $p$ be the total number of GPUs, each holding an $n$-element vector. A healthy NIC transmits one element per time unit, while a NIC with slowdown factor $\ell>1$ requires $\ell$ time units per element. Upon termination, every GPU holds the element-wise sum $\mathbf{s} = \sum_{i=1}^{p} \mathbf{x}_i$. Let $g$ be the number of GPUs per server—interconnected via NVLink with $g{-}1$ times the bandwidth of a single NIC—and let $m$ be the number of straggler servers. By the PXN mechanism (Section~\ref{sec:background}), all NICs on a given server share the same slowdown factor. We consider three scenarios: 
\begin{enumerate}[leftmargin = *]
    \item single straggler, one GPU per server ($m{=}1$, $g{=}1$);
    \item multiple stragglers, one GPU per server (fixed $m$, $g{=}1$), with slowdown factors $\ell_1 \geq \ell_2 \geq \cdots \geq \ell_m > 1$;
    \item single straggler, multiple GPUs per server (fixed $g$, $m{=}1$), where the total number of servers is $q = p/g$.
\end{enumerate}

\subsection{Best Lower Bounds}

The best known lower bounds for these three scenarios are given by Theorems~\ref{thm:tight-lower-bound}, \ref{thm:multi-lb}, and~\ref{thm:multigpu-tight-lb}, respectively; with full proofs deferred to Appendix~\ref{app:proof-lb}. We summarize these results in Table~\ref{tab:lower-bounds-summary}.

\begin{table}[ht]
\centering
\caption{Best known lower bounds on the AllReduce time 
$T$ for three straggler settings. Each bound takes the maximum of two branches; the column headers indicate the dominant branch.
Here $y_0 := 2(p{-}1)\big/(p - m + \textstyle\sum_{i=1}^{m} 1/\ell_i\big)$.}
\label{tab:lower-bounds-summary}
\vspace{4pt}
\renewcommand{\arraystretch}{1.6}
\begin{tabularx}{\textwidth}{ l @{\extracolsep{\fill}} c @{\extracolsep{\fill}} c }
\toprule
\textbf{Setting}
  & \textbf{Small-$\ell$ regime}
  & \textbf{Large-$\ell$ regime} \\
\midrule
\multirow{2}{*}{\makecell[l]{\textbf{Single straggler}\\ \textbf{One GPU per server} \\
                             \footnotesize Theorem~\ref{thm:tight-lower-bound}}}
  & $(1 < \ell < 2)$
  & $(\ell \ge 2)$ \\
  & $T \;\ge\; \dfrac{2\ell(p-1)}{\ell(p-2)+2}\cdot n$
  & $T \;\ge\; \ell \cdot n$ \\[6pt]
\midrule
\multirow{2}{*}{\makecell[l]{\textbf{Multiple stragglers} \\ \textbf{One GPU per server} \\
                             \footnotesize Theorem~\ref{thm:multi-lb}}}
  & $(1 \le \ell_1 < y_0)$
  & $(\ell_1 \ge y_0)$ \\
  & $T \;\ge\; \dfrac{2(p-1)}{p - m + \sum_{i=1}^{m}\frac{1}{\ell_i}}\cdot n$
  & $T \;\ge\; \ell_1 \cdot n$ \\[6pt]
\midrule
\multirow{2}{*}{\makecell[l]{\textbf{Single straggler}\\ \textbf{Multiple GPUs per server}\\
                             \footnotesize Theorem~\ref{thm:multigpu-tight-lb}}}
  & $(1 < \ell < 2)$
  & $(\ell \ge 2)$ \\
  & $T \;\ge\; \dfrac{2\ell(q-1)}{\ell(q-2)+2} \cdot \dfrac{n}{g}$
  & $T \;\ge\; \ell \cdot \dfrac{n}{g}$ \\[6pt]
\bottomrule
\end{tabularx}

\vspace{4pt}
\parbox{\textwidth}{\footnotesize
\textit{Notes.}
(1)~For the single-straggler cases (rows~1 and~3), the two branches
coincide at $\ell = 2$.
(2)~For $m$ stragglers, the transition point $\ell_1 = y_0$ depends
jointly on all $\ell_1 \ge \ell_2 \ge \cdots \ge \ell_m$ and $p$.
(3)~The bounds in rows~1 and~3 are \emph{tight}: matching algorithms
exist for all $p$ (resp.\ $q$); see Table~\ref{tab:alg-time-summary}.
}
\end{table}

\subsection{Algorithm Time Complexities}

The algorithms for the three scenarios are presented in Section~\ref{sec:algorithm}, Appendix~\ref{app:multi-straggler-algo}, and Appendix~\ref{app:multi-gpu}, respectively. Their time analyses appear in Section~\ref{subsec:time-analysis} (together with Appendix~\ref{app:bubble-analysis}), Appendix~\ref{app:multi-straggler-time}, and Appendix~\ref{subsec:multi-gpu-time}, respectively. We summarize these in Table~\ref{tab:alg-time-summary}.

\begin{table}[ht]
\centering
\caption{Achieved AllReduce time $T$ (as $k\to\infty$) for the three algorithm variants.
Each column shows the dominant regime; the transition occurs at $\ell = 2$ in all cases.}
\label{tab:alg-time-summary}
\vspace{4pt}
\renewcommand{\arraystretch}{1.6}
\begin{tabularx}{\textwidth}{ l @{\extracolsep{\fill}} c @{\extracolsep{\fill}} c }
\toprule
\textbf{Setting}
  & \textbf{Small-$\ell$ regime}
  & \textbf{Large-$\ell$ regime} \\
\midrule
\multirow{2}{*}{\makecell[l]{\textbf{Single straggler}\\ \textbf{One GPU per server} \\
                             \footnotesize {Section~\ref{subsec:time-analysis}, Appendix~\ref{app:bubble-analysis}}}}
  & $(1 < \ell < 2)$
  & $(\ell \ge 2)$ \\
  & $T \;\to\; \dfrac{2(p-1)\,\ell}{(p-2)\ell+2}\cdot n$
  & $T \;\to\; \ell \cdot n$ \\[6pt]
\midrule
\multirow{2}{*}{\makecell[l]{\textbf{Multiple stragglers} \\ \textbf{One GPU per server} \\
                             \footnotesize Appendix~\ref{app:multi-straggler-time}}}
  & $(1 < \ell_1 < 2)$
  & $(\ell_1 \ge 2)$ \\
  & $T \;\to\; \dfrac{2(p-1)}{p-m}\cdot n$
  & $T \;\to\; \left(\ell_1 + \dfrac{2(m-1)}{p-m}\right)\cdot n$ \\[6pt]
\midrule
\multirow{2}{*}{\makecell[l]{\textbf{Single straggler}\\ \textbf{Multiple GPUs per server}\\
                             \footnotesize Appendix~\ref{subsec:multi-gpu-time}}}
  & $(1 < \ell < 2)$
  & $(\ell \ge 2)$ \\
  & $T \;\to\; 2 \cdot \dfrac{n}{g}\, \left(\text{or } \dfrac{2\ell(q-1)}{\ell(q-2)+2} \cdot \dfrac{n}{g}\right)^{*}$
  & $T \;\to\; \ell \cdot \dfrac{n}{g}$ \\[6pt]
\bottomrule
\end{tabularx}

\vspace{4pt}
\parbox{\textwidth}{\footnotesize
$^*$ The second time complexity is achieved via the bubble-filling technique described in Appendix~\ref{subsec:multi-gpu-time}. Although theoretically possible, bubble filling is not implemented in practice.
}
\end{table}


\subsection{Optimality Results}

Combining the lower bounds from Table~\ref{tab:lower-bounds-summary} with the algorithm times from Table~\ref{tab:alg-time-summary} yields the optimality results summarized in Table~\ref{tab:optimality-summary}. Here, \textbf{optimal} means that as $k \to \infty$ with $p$, $m$, $g$, and all $\ell_i$ held constant, the algorithm achieves bandwidth optimality, confirming its optimality even for small clusters. \textbf{Near-optimal} means that as $k \to \infty$ and $p \to \infty$ with $m$, $g$, and all $\ell_i$ held constant, the algorithm approaches bandwidth optimality, confirming near-optimality in large clusters.

\begin{table}[ht]
\centering
\caption{Optimality of the three algorithm variants (as $k \to \infty$).}
\label{tab:optimality-summary}
\vspace{4pt}
\renewcommand{\arraystretch}{1.6}
\begin{tabularx}{\textwidth}{ l @{\extracolsep{\fill}} c @{\extracolsep{\fill}} c }
\toprule
\textbf{Setting}
  & \textbf{Small-$\ell$ regime}
  & \textbf{Large-$\ell$ regime} \\
\midrule
\multirow{2}{*}{\makecell[l]{\textbf{Single straggler}\\ \textbf{One GPU per server} \\
                             }}
  & $(1 < \ell < 2)$
  & $(\ell \ge 2)$ \\
  & Optimal
  & Optimal \\[6pt]
\midrule
\multirow{2}{*}{\makecell[l]{\textbf{Multiple stragglers} \\ \textbf{One GPU per server} \\
                             }}
  & $(1 < \ell_1 < 2)$
  & $(\ell_1 \ge 2)$ \\
  & Near-optimal$^\dagger$
  & Near-optimal$^\dagger$ \\[6pt]
\midrule
\multirow{2}{*}{\makecell[l]{\textbf{Single straggler}\\ \textbf{Multiple GPUs per server}\\
                             }}
  & $(1 < \ell < 2)$
  & $(\ell \ge 2)$ \\
  & Near-optimal $\left(\text{or Optimal}\right)^{*}$
  & Optimal \\[6pt]
\bottomrule
\end{tabularx}

\vspace{4pt}
\parbox{\textwidth}{\footnotesize
$^\dagger$ For Row~2, we expect that both the current algorithm and the lower bound in the small-$\ell$ regime of Table~\ref{tab:lower-bounds-summary} can be further improved, and the branching point may change accordingly. The lower bound in the large-$\ell$ regime is expected to be tight for sufficiently large $\ell_1$.\\
$^*$ The optimality result is achieved via the bubble-filling technique described in Appendix~\ref{subsec:multi-gpu-time}. Although theoretically possible, bubble filling is not implemented in practice.
}
\end{table}

\section{Proof of Information-theoretic Lower Bounds}
\label{app:proof-lb}

In this appendix, we prove all the theoretical lower bounds introduced in Section~\ref{sec:lower-bound}: Theorem~\ref{thm:lower-bound} (Appendix~\ref{app:proof-single-four-v2}), Theorem~\ref{thm:multi-lb} (Appendix~\ref{app:proof-multi}), and Theorem~\ref{thm:multigpu-lb} (Appendix~\ref{app:proof-multigpu}). We further derive tighter bounds for the single-straggler case and the multi-GPU server setting: Theorem~\ref{thm:tight-lower-bound} sharpens Theorem~\ref{thm:lower-bound}, and Theorem~\ref{thm:multigpu-tight-lb} sharpens Theorem~\ref{thm:multigpu-lb}.

\subsection{Single Straggler Case}
\label{app:proof-single-four-v2}

We first prove the bounds for the single-straggler case: Theorem~\ref{thm:lower-bound} and its tightened version, Theorem~\ref{thm:tight-lower-bound}.

Consider $p \geq 3$ GPUs performing an AllReduce on vectors of $n$ elements.
Exactly one GPU---the straggler $P_{\mathrm{s}}$---has a degraded NIC with
slowdown factor $\ell > 1$: it transmits one element in $\ell$ time units.
The remaining $p - 1$ GPUs are healthy, each transmitting one element in one
time unit. Let $F$ denote the total number of element-transfers during an
execution. Let $F_{\mathrm{s}}^{\mathrm{send}}$ and $F_{\mathrm{s}}^{\mathrm{recv}}$ denote
the total number of elements sent and received by the straggler GPU
$P_{\mathrm{s}}$.
\medskip

\noindent\textbf{Theorem~\ref{thm:lower-bound}} (restated)\textbf{.}
\textit{Any correct AllReduce algorithm satisfies}
$$
  T \;\ge\;
  \max\!\left\{\,\dfrac{2\ell(p-1)}{1+\ell(p-1)},\;\;\ell\,\right\}\cdot n\,.
$$

\begin{lemma}[Total traffic]
\label{lem:single-total-traffic-v2}
Any correct AllReduce algorithm satisfies $F \ge 2(p-1)\,n$.
\end{lemma}

\begin{proof}
Apply Lemma~\ref{lem:inter-group-2q-2} to the special case where each GPU forms its
own group, i.e., $q=p$.
\end{proof}

\begin{lemma}[Straggler incident volume]
\label{lem:single-incident-v2}
Any correct AllReduce algorithm satisfies
$F_{\mathrm{s}}^{\mathrm{send}} \ge n$ and $F_{\mathrm{s}}^{\mathrm{recv}} \ge n$.
\end{lemma}

\begin{proof}
For each coordinate $j\in[n]$, $x_{P_{\mathrm{s}},j}$ is initially known only to
$P_{\mathrm{s}}$ while every other GPU must obtain $s_j=\sum_i x_{i,j}$; hence
information about coordinate $j$ must leave $P_{\mathrm{s}}$ at least once.
Also, $P_{\mathrm{s}}$ must output $s_j$, which depends on inputs initially
outside $P_{\mathrm{s}}$; hence information about coordinate $j$ must enter
$P_{\mathrm{s}}$ at least once. Summing over $j$ gives the claim.
\end{proof}

\noindent\textit{Proof of Theorem~\ref{thm:lower-bound}.}
Let
\[
F_{\mathrm{s}}^{\mathrm{send}} = x\,n,\qquad
F_{\mathrm{s}}^{\mathrm{recv}} = y\,n,
\]
with $x\ge 1$ and $y\ge 1$ by Lemma~\ref{lem:single-incident-v2}.
Let $F_{\mathrm{h}\to\mathrm{h}}$ be the total number of element-transfers from
healthy GPUs to healthy GPUs (counted as directed transfers), and define
$F_{\mathrm{h}\to\mathrm{s}}$ and $F_{\mathrm{s}\to\mathrm{h}}$ analogously. Then
\[
F = F_{\mathrm{h}\to\mathrm{h}} + F_{\mathrm{h}\to\mathrm{s}} + F_{\mathrm{s}\to\mathrm{h}},
\quad
F_{\mathrm{h}\to\mathrm{s}} = F_{\mathrm{s}}^{\mathrm{recv}} = y\,n,
\quad
F_{\mathrm{s}\to\mathrm{h}} = F_{\mathrm{s}}^{\mathrm{send}} = x\,n.
\]

\paragraph{Straggler bottlenecks.}
Every transfer incident to $P_{\mathrm{s}}$ costs $\ell$, hence
\begin{equation}
\label{eq:single-weak-strag-send-v2}
T \ge \ell\,F_{\mathrm{s}}^{\mathrm{send}} = \ell x\,n.
\end{equation}
Similarly,
\begin{equation}
\label{eq:single-weak-strag-recv-v2}
T \ge \ell\,F_{\mathrm{s}}^{\mathrm{recv}} = \ell y\,n.
\end{equation}

\paragraph{Healthy bottlenecks.}
Ignoring the extra time cost of healthy--straggler transfers on the healthy side,
the total \emph{sending} work across the $p-1$ healthy GPUs is at least
$F_{\mathrm{h}\to\mathrm{h}} + F_{\mathrm{h}\to\mathrm{s}}$, so
\[
T \ge \frac{F_{\mathrm{h}\to\mathrm{h}} + F_{\mathrm{h}\to\mathrm{s}}}{p-1}
= \frac{F - F_{\mathrm{s}\to\mathrm{h}}}{p-1}.
\]
Using $F\ge 2(p-1)n$ (Lemma~\ref{lem:single-total-traffic-v2}) and substituting
$F_{\mathrm{s}\to\mathrm{h}}=xn$ yields
\begin{equation}
\label{eq:single-weak-healthy-send-v2}
\frac{T}{n} \ge 2 - \frac{x}{p-1}.
\end{equation}
Similarly, the total \emph{receiving} work across healthy GPUs is at least
$F_{\mathrm{h}\to\mathrm{h}} + F_{\mathrm{s}\to\mathrm{h}}$, which implies
\begin{equation}
\label{eq:single-weak-healthy-recv-v2}
\frac{T}{n} \ge 2 - \frac{y}{p-1}.
\end{equation}

\paragraph{Min--max form and symmetry reduction.}
Combining \eqref{eq:single-weak-strag-send-v2},
\eqref{eq:single-weak-strag-recv-v2},
\eqref{eq:single-weak-healthy-send-v2}, and
\eqref{eq:single-weak-healthy-recv-v2}, for all $x\ge 1,y\ge 1$,
\[
\frac{T}{n} \ge
\max\!\left\{
2 - \frac{x}{p-1},\;
2 - \frac{y}{p-1},\;
\ell x,\;
\ell y
\right\}.
\]
Let $g(x,y)$ denote the right-hand side. It is convex (a max of affine
functions) and symmetric under swapping $x$ and $y$, hence the minimum over
$x\ge 1,y\ge 1$ is attained on the diagonal $x=y=z\ge 1$. Substituting $x=y=z$,
\[
\frac{T}{n} \ge \min_{z\ge 1} \max\!\left\{\,2-\frac{z}{p-1},\;\ell z\,\right\}.
\]

\paragraph{Solving the 1-D minimization.}
The first branch is decreasing in $z$, while the second is increasing. Their
intersection is given by
\[
2-\frac{z}{p-1}=\ell z,
\]
namely
\[
z^\star=\frac{2(p-1)}{1+\ell(p-1)}.
\]
If $z^\star\ge 1$, equivalently $\ell\le 2-\frac{1}{p-1}$, then the minimum is
attained at $z^\star$ and
\[
\frac{T}{n}\ge \ell z^\star
= \frac{2\ell(p-1)}{1+\ell(p-1)}.
\]
If $z^\star<1$, equivalently $\ell> 2-\frac{1}{p-1}$, then over the feasible
region $z\ge 1$ the minimum is attained at $z=1$, which gives
\[
\frac{T}{n}\ge \max\!\left\{\,2-\frac{1}{p-1},\;\ell\,\right\}=\ell.
\]
Therefore,
\[
T \;\ge\;
n\cdot
\begin{cases}
\dfrac{2\ell(p-1)}{1+\ell(p-1)}, & 1<\ell\le 2-\dfrac{1}{p-1},\\[6pt]
\ell, & \ell> 2-\dfrac{1}{p-1}.
\end{cases}
\]
Equivalently,
\[
T \;\ge\;
\max\!\left\{\,\dfrac{2\ell(p-1)}{1+\ell(p-1)},\;\ell\,\right\}\cdot n.
\]
\hfill$\square$

\medskip
\noindent
\medskip
\noindent
\textit{Note: Theorem~\ref{thm:lower-bound} ignores the extra time cost of healthy--straggler transfers on the healthy side. Accounting for this cost yields a strictly stronger lower bound, which is in fact tight with respect to $p$; see Appendix~\ref{app:bubble-analysis} for the matching upper bound and discussion.}

\medskip

\begin{theorem}[Tight Lower Bound: Single Straggler, One GPU per Server]
\label{thm:tight-lower-bound}
Any correct AllReduce algorithm satisfies
\[
  T \;\ge\;
  \max\!\left\{\,\dfrac{2\ell(p-1)}{\ell(p-2)+2},\;\;\ell\,\right\}\cdot n\,.
\]
\end{theorem}

\noindent\textit{Proof of Theorem~\ref{thm:tight-lower-bound}.}
We follow the notation introduced in the proof of
Theorem~\ref{thm:lower-bound}. Thus,
\[
F_{\mathrm{s}}^{\mathrm{send}} = x\,n,\qquad
F_{\mathrm{s}}^{\mathrm{recv}} = y\,n,
\]
with $x\ge 1$ and $y\ge 1$, and
\[
F = F_{\mathrm{h}\to\mathrm{h}} + F_{\mathrm{h}\to\mathrm{s}} + F_{\mathrm{s}\to\mathrm{h}},
\quad
F_{\mathrm{h}\to\mathrm{s}} = y\,n,
\quad
F_{\mathrm{s}\to\mathrm{h}} = x\,n.
\]

\paragraph{Straggler bottlenecks.}
Every transfer incident to $P_{\mathrm{s}}$ costs $\ell$, hence
\begin{equation}
\label{eq:single-strag-send-v2}
T \ge \ell\,F_{\mathrm{s}}^{\mathrm{send}} = \ell x\,n.
\end{equation}
Similarly,
\begin{equation}
\label{eq:single-strag-recv-v2}
T \ge \ell\,F_{\mathrm{s}}^{\mathrm{recv}} = \ell y\,n.
\end{equation}

\paragraph{Healthy bottlenecks (time-weighted).}
Healthy$\to$healthy transfers cost $1$ per element, while any transfer between
a healthy GPU and the straggler costs $\ell$ per element. Therefore the total
\emph{sending} work across the $p-1$ healthy GPUs is at least
$F_{\mathrm{h}\to\mathrm{h}} + \ell F_{\mathrm{h}\to\mathrm{s}}$, so
\[
T \ge \frac{F_{\mathrm{h}\to\mathrm{h}} + \ell F_{\mathrm{h}\to\mathrm{s}}}{p-1}
= \frac{F - F_{\mathrm{s}\to\mathrm{h}} + (\ell-1)F_{\mathrm{h}\to\mathrm{s}}}{p-1}.
\]
Using $F\ge 2(p-1)n$ and substituting
$F_{\mathrm{s}\to\mathrm{h}}=xn$, $F_{\mathrm{h}\to\mathrm{s}}=yn$ yields
\begin{equation}
\label{eq:single-healthy-send-v2}
\frac{T}{n} \ge 2 + \frac{(\ell-1)y - x}{p-1}.
\end{equation}
Similarly, the total \emph{receiving} work across healthy GPUs is at least
$F_{\mathrm{h}\to\mathrm{h}} + \ell F_{\mathrm{s}\to\mathrm{h}}$, which implies
\begin{equation}
\label{eq:single-healthy-recv-v2}
\frac{T}{n} \ge 2 + \frac{(\ell-1)x - y}{p-1}.
\end{equation}

\paragraph{Min--max form and symmetry reduction.}
Combining \eqref{eq:single-strag-send-v2}, \eqref{eq:single-strag-recv-v2},
\eqref{eq:single-healthy-send-v2}, and \eqref{eq:single-healthy-recv-v2}, for all
$x\ge 1,y\ge 1$,
\[
\frac{T}{n} \ge
\max\!\left\{
2 + \frac{(\ell-1)y - x}{p-1},\;
2 + \frac{(\ell-1)x - y}{p-1},\;
\ell x,\;
\ell y
\right\}.
\]
Let $g(x,y)$ denote the right-hand side. It is convex (a max of affine
functions) and symmetric under swapping $x$ and $y$, hence the minimum over
$x\ge 1,y\ge 1$ is attained on the diagonal $x=y=z\ge 1$. Substituting $x=y=z$,
\[
\frac{T}{n} \ge \min_{z\ge 1} \max\!\left\{\,2+\frac{(\ell-2)z}{p-1},\;\ell z\,\right\}.
\]

\paragraph{Solving the 1-D minimization.}
If $\ell\ge 2$, both branches are nondecreasing in $z$, so the minimum is at
$z=1$ and $\frac{T}{n}\ge \ell$.
If $1<\ell\le 2$, the minimum is attained at the intersection
$2+\frac{(\ell-2)z}{p-1}=\ell z$, i.e.,
\[
z^\star=\frac{2(p-1)}{\ell(p-2)+2},
\qquad
\frac{T}{n}\ge \ell z^\star=\frac{2\ell(p-1)}{\ell(p-2)+2}.
\]
Therefore,
\[
T \;\ge\;
n\cdot
\begin{cases}
\ell, & \ell\ge 2,\\[4pt]
\dfrac{2\ell(p-1)}{\ell(p-2)+2}, & 1<\ell\le 2.
\end{cases}
\]
Equivalently,
\[
T \;\ge\;
\max\!\left\{\,\dfrac{2\ell(p-1)}{\ell(p-2)+2},\;\ell\,\right\}\cdot n.
\]
\hfill$\square$

\subsection{Multiple Stragglers Case}
\label{app:proof-multi}

Here we prove the bound for the multiple-straggler case, Theorem~\ref{thm:multi-lb}.

Consider the same setting as above, but now $m$ of the $p$ GPUs are stragglers
with heterogeneous slowdown factors
$\ell_1 \geq \ell_2 \geq \cdots \geq \ell_m > 1$; the remaining $p - m$ are
healthy. Let $F$ denote the total number of element-transfers, and let
$F_{\mathrm{s},i}^{\mathrm{send}}$ denote the total elements sent by
straggler~$i$.

\medskip
\noindent\textbf{Theorem~\ref{thm:multi-lb}} (restated)\textbf{.}
\textit{Any correct AllReduce algorithm satisfies}
$$
  T \;\geq\;
  \max\!\left\{\,\frac{2(p-1)}{p - m + \sum_{i=1}^{m}\frac{1}{\ell_i}},\;\;\ell_1\,\right\}\cdot n\,.
$$

The proof relies on the following two lemmas.

\begin{lemma}[Total traffic bound]
\label{lem:multi-total-traffic}
Any correct AllReduce algorithm satisfies $F \geq 2(p-1)\,n$.
\end{lemma}

\begin{proof}
Identical to Lemma~\ref{lem:single-total-traffic-v2}.
\end{proof}

\begin{lemma}[Straggler send volume bound]
\label{lem:multi-straggler-send}
For each straggler~$i$, any correct AllReduce algorithm satisfies
$F_{\mathrm{s},i}^{\mathrm{send}} \geq n$.
\end{lemma}

\begin{proof}
Identical to Lemma~\ref{lem:single-incident-v2}: each straggler's private data
must leave it via at least $n$ transfers.
\end{proof}

\noindent\textit{Proof of Theorem~\ref{thm:multi-lb}.}
Write $F_{\mathrm{s},i}^{\mathrm{send}} = x_i\,n$ with $x_i \geq 1$
(by Lemma~\ref{lem:multi-straggler-send}).
We derive $m + 1$ bottleneck inequalities on~$T$.

\paragraph{Straggler bottlenecks.}
Each straggler~$i$ sends $x_i\,n$ elements at rate $1/\ell_i$:
\begin{equation}
\label{eq:multi-strag-bottleneck}
  T \;\geq\; \ell_i\,x_i\,n
  \qquad \text{for each } i = 1, \ldots, m\,.
\end{equation}

\paragraph{Healthy-GPU bottleneck.}
The $p - m$ healthy GPUs collectively send at least
$F - \sum_i x_i\,n \geq 2(p-1)\,n - (\sum_i x_i)\,n$ elements
(by Lemma~\ref{lem:multi-total-traffic}). Averaging,
\begin{equation}
\label{eq:multi-healthy-bottleneck}
  T \;\geq\;
  \frac{\bigl[2(p-1) - \sum_{i=1}^{m} x_i\bigr]\,n}{p - m}\,.
\end{equation}

\paragraph{Combining.}
Taking the maximum of all $m + 1$ bottlenecks,
\[
  T \;\geq\;
  g(x_1, \ldots, x_m)
  \;:=\;
  \max\!\left\{
    \frac{\bigl[2(p-1) - \sum_{i} x_i\bigr]\,n}{p - m},\;\;
    \ell_1\,x_1\,n,\;\;
    \ldots,\;\;
    \ell_m\,x_m\,n
  \right\}.
\]
This holds for every $(x_1, \ldots, x_m)$ with $x_i \geq 1$. Define
\[
  y_0
  \;:=\;
  \frac{2(p-1)}{p - m + \sum_{i=1}^{m} \dfrac{1}{\ell_i}}
\]
and $x_i^{\star} := y_0 / \ell_i$.
These are obtained by equating all $m+1$ terms of $g$: setting
$\ell_i\,x_i = y$ for all~$i$ (giving $x_i = y/\ell_i$) and solving
$[2(p-1) - y\sum_i 1/\ell_i]/(p-m) = y$ yields $y = y_0$.

The condition $x_i^{\star} \geq 1$ for all~$i$ reduces to
$y_0 \geq \ell_1$ (since $\ell_1 \geq \ell_i$, if
$x_1^{\star} = y_0/\ell_1 \geq 1$ then $x_i^{\star} \geq 1$ for all~$i$).

\smallskip
\noindent\emph{Case~1 ($y_0 \geq \ell_1$, i.e., all $x_i^{\star} \geq 1$):}
The minimum of $g$ on $\{x_i \geq 1\}$ is attained at
$(x_1^{\star}, \ldots, x_m^{\star})$ with value $y_0\,n$.

\smallskip
\noindent\emph{Case~2 ($y_0 < \ell_1$, i.e., $x_1^{\star} < 1$):}
The balanced solution is infeasible. Since $g$ is a max of
terms including $\ell_1\,x_1\,n \geq \ell_1\,n$, the minimum of $g$
on $\{x_i \geq 1\}$ is at least $\ell_1\,n$.

\smallskip
Combining both cases,
\[
  T \;\geq\; \min_{x_i\ge 1}\, g(x_1,\ldots,x_m)
  \;=\;
  \max\!\left\{
    \frac{2(p-1)}{p - m + \sum_{i=1}^{m}\frac{1}{\ell_i}},\;\;\ell_1
  \right\}\cdot n\,,
\]
completing the proof. \hfill$\square$

\begin{remark}
Unlike the other two subsections in this section, the refinement of the
healthy-side bottleneck does not extend directly to the present setting with
multiple bad servers. The reason is that, once the inter-server traffic is
decomposed according to healthy and bad endpoints, one now obtains four directed
classes,
\[
F_{\mathrm{h}\to\mathrm{h}},\qquad
F_{\mathrm{h}\to\mathrm{b}},\qquad
F_{\mathrm{b}\to\mathrm{h}},\qquad
F_{\mathrm{b}\to\mathrm{b}},
\]
rather than three. In particular, the additional bad-to-bad traffic
$F_{\mathrm{b}\to\mathrm{b}}$ is not captured, so the
same strengthening trick used in the previous two subsections cannot be applied
here in a direct way.
\end{remark}

\paragraph{Consistency with the single-straggler bound.}
Setting $m = 1$ recovers
$T \geq \max\{2\ell(p{-}1)/(\ell(p{-}1)+1),\;\ell\}\cdot n$,
which coincides with Theorem~\ref{thm:lower-bound}.

\subsection{Multi-GPU per Server Case}
\label{app:proof-multigpu}

Finally, we prove the bounds for the multi-GPU server setting with a single straggler: Theorem~\ref{thm:multigpu-lb} and its tightened version, Theorem~\ref{thm:multigpu-tight-lb}.

Write $q:=p/g$ for the number of servers. We treat each server as a supernode.
Inter-server communication uses NICs; intra-server communication uses NVLink.
After normalization, each healthy NIC has unit rate, so each healthy server has
aggregate NIC capacity $g$. One server $S_s$ is slow: each of its NICs is degraded
by factor $\ell$, hence $S_s$ has aggregate NIC capacity $g/\ell$.
Let $F$ be the total traffic, $F_{\mathrm{nic}}$ the inter-server NIC traffic,
$\Fsend$ the NIC traffic sent by $S_s$, and $F_{\mathrm{nv}}=F-F_{\mathrm{nic}}$
the NVLink traffic. (We do not need a lower bound on $F_{\mathrm{nv}}$; ignoring
NVLink can only weaken a lower bound on the runtime.)

\paragraph{NIC-time model.}
We count \emph{directed} element-transfers. Any inter-server element-transfer that
is incident to $S_s$ takes $\ell$ time units; any inter-server transfer between
healthy servers takes $1$ time unit. Each server can send/receive at aggregate
rate $g$ (healthy) or $g/\ell$ (slow). Therefore, the total time $T$ must be at
least the corresponding (time-weighted) NIC workload divided by the relevant
aggregate capacity.

\medskip
\noindent\textbf{Theorem~\ref{thm:multigpu-lb}} (restated)\textbf{.}
\textit{Any correct AllReduce algorithm satisfies}
\[
  T \;\ge\; \frac{n}{g}\cdot
  \max\!\left\{
  \frac{2\ell(q-1)}{1+\ell(q-1)},\;
  \ell
  \right\}.
\]

\begin{lemma}[Inter-group traffic per coordinate via supernodes]
\label{lem:inter-group-2q-2}
Partition the $p$ GPUs into $q\ge 1$ disjoint groups
$\mathcal{G}_1,\dots,\mathcal{G}_q$.
Count only \emph{inter-group element-transfers}: a transfer contributes $1$ to
coordinate $j$ iff it carries the $j$-th coordinate from some GPU in
$\mathcal{G}_u$ to some GPU in $\mathcal{G}_v$ with $u\neq v$.

Then any correct AllReduce algorithm satisfies that for every coordinate $j\in[n]$,
the number of inter-group element-transfers carrying coordinate $j$ is at least
$2(q-1)$. Consequently, the total inter-group traffic satisfies
$F_{\mathrm{inter}} \ge 2(q-1)n$.
\end{lemma}

\begin{proof}
Fix a coordinate $j\in[n]$ and let $s_j=\sum_{i=1}^p x_{i,j}$.

\paragraph{Supernode flow graph.}
For each time $t$, define a directed multigraph $H_{\le t}^{(j)}$ on vertex set
$[q]=\{1,\dots,q\}$ (the groups as \emph{supernodes}) as follows:
for every inter-group element-transfer of coordinate $j$ that completes by time $t$
from a sender in group $\mathcal{G}_u$ to a receiver in group $\mathcal{G}_v$,
add a directed edge $u\to v$.

For a group index $r\in[q]$ and time $t$, say that \emph{group $r$ can compute $s_j$
at time $t$} if there exists some GPU in $\mathcal{G}_r$ whose local state at time $t$
uniquely determines $s_j$ (equivalently, there exists a decoding function of its
local state that outputs $s_j$ for all inputs).

Let
\[
t_j \;=\; \inf\{t:\exists r\in[q]\ \text{such that group $r$ can compute $s_j$ at time $t$}\},
\]
and define the set of \emph{roots at the first completion time}
\[
A_j \;=\; \{r\in[q]: \text{group $r$ can compute $s_j$ at time } t_j\},
\qquad a_j=|A_j|\ (\ge 1).
\]

\paragraph{Key reachability property.}
Fix $r\in A_j$. Then every group $u\in[q]$ must have a directed path $u\leadsto r$
in $H_{\le t_j}^{(j)}$; otherwise the state within group $r$ at time $t_j$ is
independent of the inputs from $\mathcal{G}_u$ (for coordinate $j$), so by an
indistinguishability argument $r$ cannot uniquely determine $s_j$ at time $t_j$.
In particular, for any chosen $r_0\in A_j$, all $u$ reach $r_0$.
Choose an arbitrary $r_0\in A_j$. Since all vertices reach $r_0$ in
$H_{\le t_j}^{(j)}$, the graph contains a spanning in-arborescence
$T_{\mathrm{in}}$ rooted at $r_0$, hence
\[
|E(T_{\mathrm{in}})|\ge q-1.
\]

\paragraph{``Last-hop at $t_j$'' into each root and disjointness.}
For every $r\in A_j$, there exists an inter-group edge $e_r$ completing
\emph{exactly} at time $t_j$ whose head is $r$; otherwise group $r$ would receive
no new inter-group information about coordinate $j$ at time $t_j$, so $r$ would
already be able to compute $s_j$ at some $t<t_j$, contradicting minimality of $t_j$.
The edges $\{e_r\}_{r\in A_j}$ are pairwise distinct (different heads).

Any
edge of $T_{\mathrm{in}}$ whose head is not $r_0$ must complete at time $<t_j$
(otherwise information could not continue onward to reach $r_0$ by time $t_j$).
Therefore,
\[
\{e_r:\ r\in A_j\setminus\{r_0\}\}\cap E(T_{\mathrm{in}})=\emptyset.
\]
Hence the number of inter-group $j$-edges completed by time $t_j$ satisfies
\[
|E(H_{\le t_j}^{(j)})|
\ \ge\ 
|E(T_{\mathrm{in}})| + \bigl|\{e_r:\ r\in A_j\setminus\{r_0\}\}\bigr|
\ \ge\ 
(q-1)+(a_j-1).
\]

\paragraph{Edges after $t_j$.}
For any group $v\notin A_j$, correctness implies it must eventually compute $s_j$.
If $v$ received no inter-group $j$-edge after $t_j$, its inter-group information
about coordinate $j$ would never change after $t_j$, so it could never become able
to compute $s_j$, a contradiction. Thus each $v\notin A_j$ must be the head of at
least one inter-group $j$-edge completing after $t_j$. These edges are distinct
across different $v$, giving at least $q-a_j$ further edges.

\paragraph{Per-coordinate bound and summation.}
Let $F_{\mathrm{inter}}^{(j)}$ be the number of inter-group element-transfers
carrying coordinate $j$. Combining,
\[
F_{\mathrm{inter}}^{(j)}\ \ge\ \bigl((q-1)+(a_j-1)\bigr) + (q-a_j)\ =\ 2(q-1).
\]
Summing over $j\in[n]$ yields $F_{\mathrm{inter}}\ge 2(q-1)n$.
\end{proof}

\begin{lemma}[Inter-server NIC traffic]
\label{lem:multigpu-Fnic}
$F_{\mathrm{nic}} \ge 2(q-1)n$.
\end{lemma}

\begin{proof}
Apply Lemma~\ref{lem:inter-group-2q-2} with the $q$ servers as groups (supernodes).
\end{proof}

\begin{lemma}[Slow-server incident NIC volume]
\label{lem:multigpu-incident}
Let $F_s^{\mathrm{send}}:=\Fsend$ be the total NIC traffic sent by $S_s$, and
let $F_s^{\mathrm{recv}}$ be the total NIC traffic received by $S_s$. Then
\[
F_s^{\mathrm{send}} \ge n,
\qquad
F_s^{\mathrm{recv}} \ge n.
\]
\end{lemma}

\begin{proof}
Fix a coordinate $j\in[n]$ and write $s_j=\sum_{i=1}^p x_{i,j}$.

If no inter-server transfer carrying coordinate $j$ is ever sent from $S_s$, then
all other servers' views (restricted to inter-server information) are independent
of the inputs in $S_s$ for coordinate $j$, so by indistinguishability they cannot
uniquely determine $s_j$ for all inputs, contradicting correctness. Hence for each
$j$, at least one $j$-element must be sent from $S_s$, giving $F_s^{\mathrm{send}}\ge n$.

Similarly, we have $F_s^{\mathrm{recv}}\ge n$.
\end{proof}

\noindent\textit{Proof of Theorem~\ref{thm:multigpu-lb}.}
Decompose the directed NIC traffic as
\[
F_{\mathrm{nic}}
= F_{\mathrm{h}\to\mathrm{h}} + F_{\mathrm{h}\to s} + F_{s\to\mathrm{h}},
\qquad
F_{s\to\mathrm{h}} = F_s^{\mathrm{send}},\quad F_{\mathrm{h}\to s}=F_s^{\mathrm{recv}}.
\]
Let
\[
F_s^{\mathrm{send}} = x\,n,\qquad F_s^{\mathrm{recv}} = y\,n,
\]
so $x\ge 1$ and $y\ge 1$ by Lemma~\ref{lem:multigpu-incident}.

\paragraph{Slow-server NIC bottlenecks.}
Since $S_s$ has aggregate NIC rate $g/\ell$,
\begin{equation}
\label{eq:multigpu-weak-slow-send}
T \ge \frac{F_s^{\mathrm{send}}}{g/\ell} = \frac{\ell x\,n}{g}.
\end{equation}
Similarly,
\begin{equation}
\label{eq:multigpu-weak-slow-recv}
T \ge \frac{F_s^{\mathrm{recv}}}{g/\ell} = \frac{\ell y\,n}{g}.
\end{equation}

\paragraph{Healthy-server NIC bottlenecks.}
Ignoring the extra time cost of healthy--slow transfers on the healthy side,
the total \emph{sending} work across the $q-1$ healthy servers is at least
$F_{\mathrm{h}\to\mathrm{h}} + F_{\mathrm{h}\to s}$, so
\[
T \ge \frac{F_{\mathrm{h}\to\mathrm{h}} + F_{\mathrm{h}\to s}}{(q-1)g}
= \frac{F_{\mathrm{nic}} - F_{s\to\mathrm{h}}}{(q-1)g}.
\]
Using Lemma~\ref{lem:multigpu-Fnic} and substituting $F_{s\to\mathrm{h}}=xn$ gives
\begin{equation}
\label{eq:multigpu-weak-healthy-send}
\frac{T}{n} \ge \frac{1}{g}\left(2 - \frac{x}{q-1}\right).
\end{equation}
Similarly, the total \emph{receiving} work across healthy servers is at least
$F_{\mathrm{h}\to\mathrm{h}} + F_{s\to\mathrm{h}}$, hence
\begin{equation}
\label{eq:multigpu-weak-healthy-recv}
\frac{T}{n} \ge \frac{1}{g}\left(2 - \frac{y}{q-1}\right).
\end{equation}

\paragraph{Min--max and symmetry.}
Combining \eqref{eq:multigpu-weak-slow-send},
\eqref{eq:multigpu-weak-slow-recv},
\eqref{eq:multigpu-weak-healthy-send}, and
\eqref{eq:multigpu-weak-healthy-recv}, for all $x,y\ge 1$,
\[
\frac{T}{n} \ge \frac{1}{g}\max\!\left\{
2-\frac{x}{q-1},\;
2-\frac{y}{q-1},\;
\ell x,\;
\ell y
\right\}.
\]
By symmetry in $x,y$ (and convexity of the max of affine functions), the minimum
over $x,y\ge 1$ is attained at $x=y=z\ge 1$, yielding
\[
\frac{T}{n} \ge \frac{1}{g}\min_{z\ge 1}\max\!\left\{\,2-\frac{z}{q-1},\ \ell z\,\right\}.
\]

\paragraph{Solving the 1-D minimization.}
The two branches intersect at
\[
2-\frac{z}{q-1}=\ell z,
\]
namely
\[
z^\star=\frac{2(q-1)}{1+\ell(q-1)}.
\]
If $z^\star\ge 1$, equivalently $\ell\le 2-\frac{1}{q-1}$, then the minimum is
attained at $z^\star$ and
\[
\frac{T}{n}\ge \frac{\ell z^\star}{g}
= \frac{1}{g}\cdot \frac{2\ell(q-1)}{1+\ell(q-1)}.
\]
If $z^\star<1$, equivalently $\ell> 2-\frac{1}{q-1}$, then over the feasible
region $z\ge 1$ the minimum is attained at $z=1$, which gives
\[
\frac{T}{n}\ge \frac{1}{g}\max\!\left\{\,2-\frac{1}{q-1},\;\ell\,\right\}
= \frac{\ell}{g}.
\]
Therefore,
\[
T \;\ge\; \frac{n}{g}\cdot
\begin{cases}
\dfrac{2\ell(q-1)}{1+\ell(q-1)}, & 1<\ell\le 2-\dfrac{1}{q-1},\\[6pt]
\ell, & \ell> 2-\dfrac{1}{q-1}.
\end{cases}
\]
Equivalently,
\[
T \;\ge\; \frac{n}{g}\cdot
\max\!\left\{
\frac{2\ell(q-1)}{1+\ell(q-1)},\;
\ell
\right\}.
\]
\hfill$\square$

\medskip
\noindent
Similarly, \textit{Theorem~\ref{thm:multigpu-lb} ignores the extra time cost of healthy--slow transfers on the healthy side. Accounting for this cost yields a strictly stronger lower bound, which is in fact tight; see Appendix~\ref{app:multi-gpu} for the matching upper bound and discussion.}
\medskip

\begin{theorem}[Tight Lower Bound: Single Straggler, Multiple GPU per Server]
\label{thm:multigpu-tight-lb}
Any correct AllReduce algorithm satisfies
\[
  T \;\ge\; \frac{n}{g}\cdot
  \max\!\left\{
  \frac{2\ell(q-1)}{\ell(q-2)+2},\;
  \ell
  \right\}.
\]
\end{theorem}

\noindent\textit{Proof of Theorem~\ref{thm:multigpu-tight-lb}.}
We follow the notation introduced in the proof of
Theorem~\ref{thm:multigpu-lb}. Thus,
\[
F_{\mathrm{nic}}
= F_{\mathrm{h}\to\mathrm{h}} + F_{\mathrm{h}\to s} + F_{s\to\mathrm{h}},
\qquad
F_{s\to\mathrm{h}} = F_s^{\mathrm{send}},\quad F_{\mathrm{h}\to s}=F_s^{\mathrm{recv}},
\]
and
\[
F_s^{\mathrm{send}} = x\,n,\qquad F_s^{\mathrm{recv}} = y\,n,
\]
with $x\ge 1$ and $y\ge 1$.

\paragraph{Slow-server NIC bottlenecks.}
Since $S_s$ has aggregate NIC rate $g/\ell$,
\begin{equation}
\label{eq:multigpu-slow-send}
T \ge \frac{F_s^{\mathrm{send}}}{g/\ell} = \frac{\ell x\,n}{g}.
\end{equation}
Similarly,
\begin{equation}
\label{eq:multigpu-slow-recv}
T \ge \frac{F_s^{\mathrm{recv}}}{g/\ell} = \frac{\ell y\,n}{g}.
\end{equation}

\paragraph{Healthy-server NIC bottlenecks (time-weighted).}
Across the $q-1$ healthy servers, the aggregate send rate is $(q-1)g$. Sending
healthy$\to$healthy elements costs $1$ time unit each, while any element sent to
(or received from) $S_s$ costs $\ell$ time units (the slow NIC throttles the transfer).
Thus the total \emph{sending} work across healthy servers is at least
$F_{\mathrm{h}\to\mathrm{h}} + \ell F_{\mathrm{h}\to s}$, implying
\[
T \ge \frac{F_{\mathrm{h}\to\mathrm{h}}+\ell F_{\mathrm{h}\to s}}{(q-1)g}
= \frac{F_{\mathrm{nic}} - F_{s\to\mathrm{h}} + (\ell-1)F_{\mathrm{h}\to s}}{(q-1)g}.
\]
Using Lemma~\ref{lem:multigpu-Fnic} and substituting $F_{s\to\mathrm{h}}=xn$,
$F_{\mathrm{h}\to s}=yn$ gives
\begin{equation}
\label{eq:multigpu-healthy-send}
\frac{T}{n} \ge \frac{1}{g}\left(2 + \frac{(\ell-1)y-x}{q-1}\right).
\end{equation}
Similarly, the total \emph{receiving} work across healthy servers is at least
$F_{\mathrm{h}\to\mathrm{h}} + \ell F_{s\to\mathrm{h}}$, hence
\begin{equation}
\label{eq:multigpu-healthy-recv}
\frac{T}{n} \ge \frac{1}{g}\left(2 + \frac{(\ell-1)x-y}{q-1}\right).
\end{equation}

\paragraph{Min--max and symmetry.}
Combining \eqref{eq:multigpu-slow-send}, \eqref{eq:multigpu-slow-recv},
\eqref{eq:multigpu-healthy-send}, \eqref{eq:multigpu-healthy-recv}, for all $x,y\ge 1$,
\[
\frac{T}{n} \ge \frac{1}{g}\max\!\left\{
2+\frac{(\ell-1)y-x}{q-1},\;
2+\frac{(\ell-1)x-y}{q-1},\;
\ell x,\;
\ell y
\right\}.
\]
By symmetry in $x,y$ (and convexity of the max of affine functions), the minimum
over $x,y\ge 1$ is attained at $x=y=z\ge 1$, yielding
\[
\frac{T}{n} \ge \frac{1}{g}\min_{z\ge 1}\max\!\left\{\,2+\frac{(\ell-2)z}{q-1},\ \ell z\,\right\}.
\]

\paragraph{Solving the 1-D minimization.}
If $\ell\ge 2$, both branches are nondecreasing in $z$, so the minimum is at
$z=1$ and $\frac{T}{n}\ge \frac{\ell}{g}$.
If $1<\ell\le 2$, the minimum is attained at the intersection
\[
2+\frac{(\ell-2)z}{q-1}=\ell z,
\]
i.e.,
\[
z^\star=\frac{2(q-1)}{\ell(q-2)+2},
\qquad
\frac{T}{n}\ge \frac{\ell z^\star}{g}
= \frac{1}{g}\cdot \frac{2\ell(q-1)}{\ell(q-2)+2}.
\]
Therefore,
\[
T \;\ge\; \frac{n}{g}\cdot
\begin{cases}
\ell, & \ell\ge 2,\\[6pt]
\dfrac{2\ell(q-1)}{\ell(q-2)+2}, & 1<\ell\le 2.
\end{cases}
\]
Equivalently,
\[
T \;\ge\; \frac{n}{g}\cdot
\max\!\left\{
\frac{2\ell(q-1)}{\ell(q-2)+2},\;
\ell
\right\}.
\]
\hfill$\square$

\begin{remark}[Dropping the NVLink term]
In addition to the NIC constraints, we also have the NVLink-time constraint
\[
T \;\ge\; \frac{F_{\mathrm{nv}}}{qv},
\]
since the $q$ servers have total aggregate NVLink capacity $qv$. Using the global
traffic bound $F\ge 2(p-1)n$ and $F_{\mathrm{nv}}=F-F_{\mathrm{nic}}$, we get
\[
F_{\mathrm{nv}} \;\ge\; 2(p-1)n - F_{\mathrm{nic}}.
\]
Plugging in the minimum possible inter-server traffic $F_{\mathrm{nic}}=2(q-1)n$
(from Lemma~\ref{lem:inter-group-2q-2}) yields
\[
T \;\ge\; \frac{2(p-1)n-2(q-1)n}{qv}
\;=\; \frac{2(p-q)n}{qv}
\;=\; \frac{2(g-1)n}{v},
\]
where we used $p=gq$. Under our NVLink-rich assumption $v\ge g(g-1)$, this term is
at most $2n/g$, and hence it does not affect the main results of this section.

Intuitively, in \emph{NVLink-poor} settings (e.g., per-GPU NVLink bandwidth is
\emph{not} at least $(g-1)$ times the NIC bandwidth, i.e., $v<g(g-1)$ in our
normalization), then $F_{\mathrm{nv}}/(qv)$ can become the dominant bottleneck and
should be kept explicitly. Conversely, when NVLink is at least this $(g-1)$-fold
``multiplier,'' intra-server aggregation is never the limiter. 
\end{remark}

\section{Time Analysis for Bubble Filling (\texorpdfstring{$\ell < 2$}{l < 2})}
\label{app:bubble-analysis}

When $\ell < 2$, each Stage~2/3 flow occupies $\ell\,s$ time on the
straggler link while the corresponding ring slot takes $2s$.
Bubble filling enlarges every Stage~2/3 flow to carry an additional
$(2{-}\ell)\,s'/\ell$ elements of a point-to-point AllReduce between
the healthy servers and the straggler, where $s'$ denotes the new
section size.  The enlarged flow has duration
$(1 + (2{-}\ell)/\ell)\cdot\ell\,s' = 2s'$, matching the ring slot
exactly.

\paragraph{Data accounting.}
Each group of four segments (one per pattern) processes two kinds of data:
\begin{itemize}
  \item \textbf{Ring path:} $(p{-}1)$ sections of size $s'$, totalling
        $(p{-}1)\,s'$ elements.
  \item \textbf{P2P path:} $(2{-}\ell)\,s'/\ell$ additional elements
        allreduced via the straggler bubbles.
\end{itemize}
With $k/4$ groups (i.e.\ $k$ segments total) the constraint
$n = k\bigl[(p{-}1)\,s' + (2{-}\ell)\,s'/\ell\bigr]$ gives
\[
  s' \;=\; \frac{n\,\ell}{k\bigl[(p{-}2)\ell + 2\bigr]}\,.
\]

\paragraph{Total time.}
Because every slot now has duration~$2s'$, each middle body takes
$2(p{-}1)\,s'$.  The straggler's serialized recv chain across all
$k{+}1$ bodies takes $2k(p{-}1)\,s'$ (the HEAD bub-up and TAIL Stage 2
flows are shorter than middle Stage 2 flows, but their durations sum to
exactly $2k(p{-}1)\,s'$).  In addition, ring-flow dependencies at the
HEAD and TAIL add a startup/drain overhead of
$2(\ell{-}1)(p{-}1)\,s'$ (at $\ell=1$ the HEAD/TAIL flows are short
enough to hide entirely within the straggler chain; at $\ell=2$ they
contribute a full extra body).  Hence
\begin{align*}
  T &= \bigl[2k + 2(\ell{-}1)\bigr]\,(p{-}1)\,s'
     = \frac{2(p{-}1)\,\ell\,n}{(p{-}2)\ell+2}
       \,\frac{k+\ell-1}{k}\\
    &\xrightarrow{k\to\infty}\;
     \frac{2(p{-}1)\,\ell\,n}{(p{-}2)\ell+2}\,.
\end{align*}

Comparing with Theorem~\ref{thm:tight-lower-bound}, we see that our algorithm \emph{achieves bandwidth optimality without requiring $p \to \infty$}, confirming its optimality even for small clusters.

\paragraph{Overhead bound.}
Compared with the non-degraded baseline $T_0 = 2(p{-}1)\,n/p$ (single
GPU per server, $g{=}1$),
\[
  \frac{T}{T_0}
  = \frac{p\,\ell}{(p{-}2)\ell+2}\,.
\]
This ratio is increasing in~$\ell$ on $[1,2]$ and attains its maximum
$p/(p{-}1) = 1 + 1/(p{-}1)$ at $\ell=2$, so the overhead is at most
$1+1/(p{-}1)$ for all $\ell\in[1,2]$.

\section{Algorithm Construction for Multiple Stragglers}
\label{app:multi-straggler-algo}

We extend the four-stage pipelined algorithm of
Section~\ref{sec:algorithm} from one straggler to $m \geq 2$
stragglers. Throughout, $p$ denotes the total number of GPUs (one per
server), of which GPUs~$0, \ldots, m{-}1$ are stragglers with
slowdown factors $\ell_1 \geq \cdots \geq \ell_m > 1$, and
GPUs~$m, \ldots, p{-}1$ are healthy.

\subsection{Four-Stage Decomposition}

As in Section~\ref{subsec:four-stage}, we partition the $n$ input
elements into $k$ segments, each containing $p{-}m$ sections of
size $s = n / (k(p{-}m))$. Each segment undergoes four stages:

\begin{enumerate}[leftmargin = *]
  \item \textbf{Stage~1 (Reduce-Scatter):} The $p{-}m$ healthy GPUs
        perform a reduce-scatter of the segment along a directed ring,
        accumulating partial sums over $p{-}m{-}1$ hops.
  \item \textbf{Stage~2 (Upload):} Each healthy GPU sends its partial
        sum to \emph{every} straggler. Each straggler receives $p{-}m$
        uploads over its slow NIC.
  \item \textbf{Stage~3 (Download):} Each straggler folds in its
        local contribution and sends the result back to every healthy
        GPU. Each straggler sends $p{-}m$ downloads.
  \item \textbf{Stage~4 (Allgather):} The healthy GPUs allgather the
        global sums along the ring in $p{-}m{-}1$ hops.
\end{enumerate}

The same two stage orderings from
Section~\ref{par:stage-ordering} remain valid, and the
\emph{disjoint-resource} principle still holds: Stages~2/3 use only
the straggler NICs while Stages~1/4 use only the healthy ring links.

\subsection{Schedule Construction}

We design four patterns---B, D, A$'$, C$'$---analogous to the
single-straggler patterns of Section~\ref{subsec:schedule}, but
adapted to multiple stragglers. We illustrate the construction for
$p = 7$, $m = 2$ (GPUs~0--1 are stragglers; GPUs~2--6 are healthy).

\paragraph{Window structure.}
Each parallel body is divided into $p{-}1$ windows, each comprising
two sub-slots. Of these, $p{-}m{-}1$ windows carry merged ring flows
(S1+S4 or S1+S3, as in the single-straggler case), and $m$ windows
carry straggler communication (one per straggler). The body
duration is $2(p{-}1)\,s$ when all straggler flows fit inside a
$2s$-wide window, which holds whenever $\ell_i \leq 2$ for all~$i$.

\paragraph{Four patterns.}
The patterns are defined as follows:
\begin{itemize}[leftmargin = *]
  \item \textbf{Pattern~B}: S3 $\to$ S1 $\to$ S4 $\to$ S2.
  \item \textbf{Pattern~D}: S3 $\to$ S1 $\to$ S4 $\to$ S2
        (same ordering as~B, with the same offset).
  \item \textbf{Pattern~A$'$}: S1 $\to$ S2 $\to$ S3 $\to$ S4.
  \item \textbf{Pattern~C$'$}: S1 $\to$ S2 $\to$ S3 $\to$ S4
        (same ordering as~A$'$, with flows in Stages~1 and~3 shifted).
\end{itemize}
Figure~\ref{fig:ms-four-patterns} shows the four individual patterns
for $p = 7$, $m = 2$.

\begin{figure*}[t]
  \centering
  \begin{subfigure}[b]{0.48\textwidth}
    \centering
    \includegraphics[width=\textwidth]{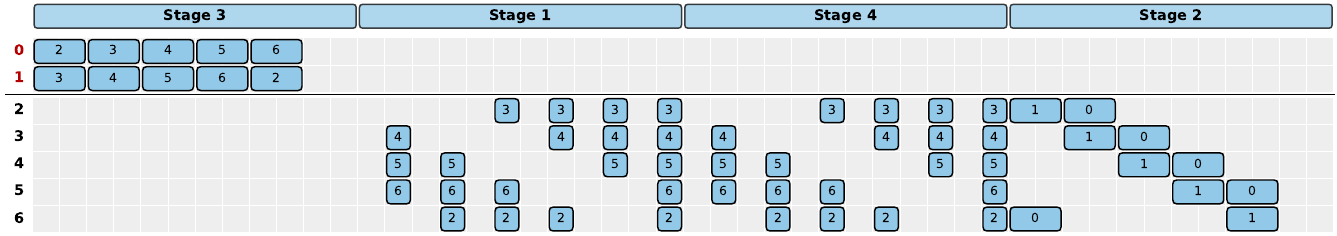}
    \caption{Pattern~B (S3$\to$S1$\to$S4$\to$S2)}
    \label{fig:ms-pattern-b}
  \end{subfigure}\hfill
  \begin{subfigure}[b]{0.48\textwidth}
    \centering
    \includegraphics[width=\textwidth]{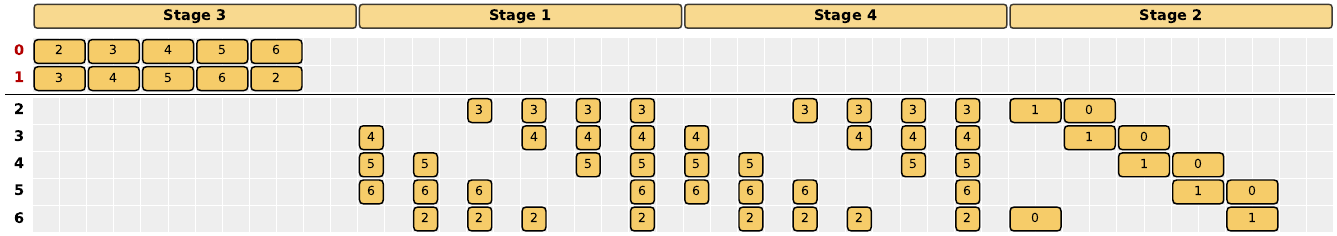}
    \caption{Pattern~D (S3$\to$S1$\to$S4$\to$S2)}
    \label{fig:ms-pattern-d}
  \end{subfigure}\\[6pt]
  \begin{subfigure}[b]{0.48\textwidth}
    \centering
    \includegraphics[width=\textwidth]{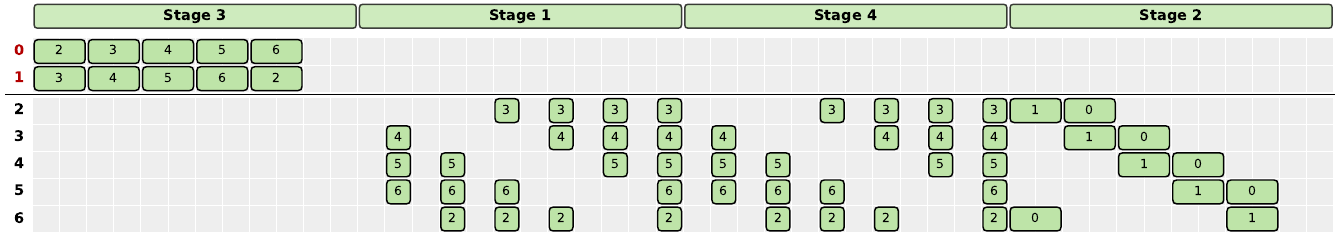}
    \caption{Pattern~A$'$ (S1$\to$S2$\to$S3$\to$S4)}
    \label{fig:ms-pattern-aprime}
  \end{subfigure}\hfill
  \begin{subfigure}[b]{0.48\textwidth}
    \centering
    \includegraphics[width=\textwidth]{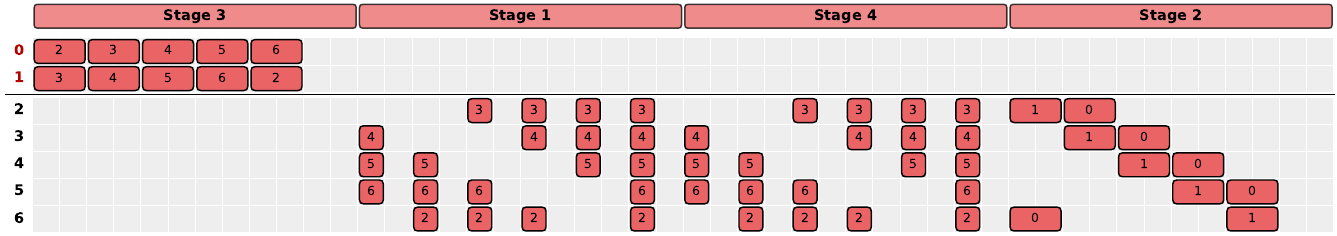}
    \caption{Pattern~C$'$ (S1$\to$S2$\to$S3$\to$S4)}
    \label{fig:ms-pattern-cprime}
  \end{subfigure}
  \caption{Flow schedules for the four multi-straggler patterns
    ($p{=}7$, $m{=}2$). Rows 0--1 (red labels) are stragglers;
    rows 2--6 are healthy. Each cell's label indicates the destination
    GPU. Colored cells belong to the named pattern; uncolored cells
    belong to other patterns.}
  \label{fig:ms-four-patterns}
\end{figure*}

\paragraph{Pipelining.}
Overlaying the four patterns produces the composite schedule of
Figure~\ref{fig:ms-patterns-combined}. With $k$ a multiple of~4,
each pattern processes $k/4$ segments. Every parallel body in the
steady state contains one instance of each stage from each pattern,
fully utilizing both the healthy ring and all straggler links.

\begin{figure*}[t]
  \centering
  \includegraphics[width=\textwidth]{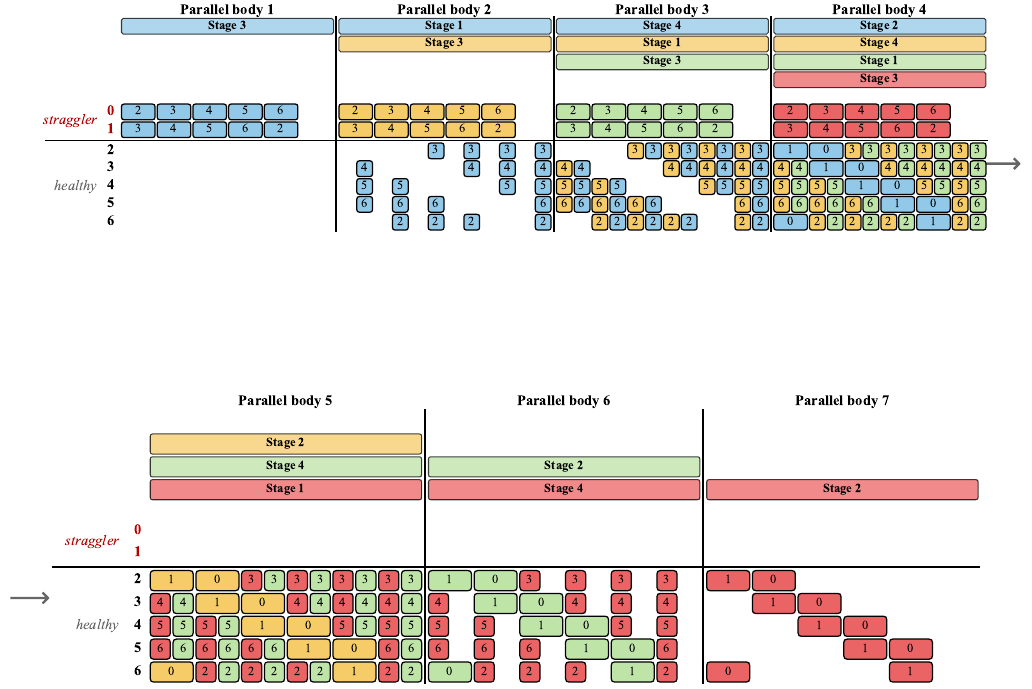}
  \caption{Combining the four multi-straggler patterns into a
    complete pipeline ($p{=}7$, $m{=}2$). Rows 0--1 are stragglers;
    rows 2--6 are healthy. The schedule has 7 parallel bodies; each
    steady-state body interleaves all four stages across the four
    patterns. No NIC sends or receives two flows simultaneously.}
  \label{fig:ms-patterns-combined}
\end{figure*}

\subsection{Time Analysis}
\label{app:multi-straggler-time}
Each segment has $p{-}m$ sections of size
$s = n/(k(p{-}m))$.
In each parallel body, the healthy GPUs occupy $p{-}1$ windows
of $2s$ each, for a total of $2(p{-}1)\,s$.
Each straggler~$i$ must send or receive $p{-}m$ flows, each of
duration $\ell_i\,s$, giving a straggler chain of
$\ell_i(p{-}m)\,s$. In the current schedule, additional $2(m-1)$ sections between healthy servers will be sent after the healthy servers finish communication with stragglers. The body duration is therefore
\[
  T_{\mathrm{body}}
  \;=\;
  \max\!\left\{\,2(p{-}1)\,s,\;\;
  \left(\ell_1(p{-}m) + 2(m-1)\right)\,s\,\right\}\,.
\]
With $k/4$ groups (each spanning four patterns) plus startup/drain
equivalent to one body, the total time is
\[
  T \;=\; \left(\frac{k}{4} + 1\right)\cdot 4\,T_{\mathrm{body}}
  \;=\; (k+4)\,T_{\mathrm{body}}\,.
\]

\paragraph{Case $\ell_1 < 2$ (ring-bottlenecked).}
The body is determined by the healthy GPUs:
$T_{\mathrm{body}} = 2(p{-}1)\,s$. Substituting $s$:
\begin{align*}
  T &= (k+4)\cdot 2(p{-}1)\,s
     = \frac{2(p{-}1)}{p{-}m}\,n\,\frac{k+4}{k}\\
    &\xrightarrow{k\to\infty}\;
     \frac{2(p{-}1)}{p{-}m}\,n\,.
\end{align*}

\paragraph{Case $\ell_1 \ge 2$ (straggler-bottlenecked).}
The body is determined by the worst straggler:
$T_{\mathrm{body}} = \left(\ell_1(p{-}m) + 2(m-1)\right)\,s$. Substituting:
\begin{align*}
  T &= (k+4)\cdot\left(\ell_1(p{-}m) + 2(m-1)\right)\,s
     = \frac{\ell_1(p{-}m) + 2(m-1)}{p-m}\,n\,\frac{k+4}{k}\\
    &\xrightarrow{k\to\infty}\;
     \left(\ell_1 + \frac{2(m-1)}{p-m}\right)\,n\,.
\end{align*}

\paragraph{Comparison with the lower bound.}
From Theorem~\ref{thm:multi-lb}, the lower bound (for $p \gg m$) is
$T \geq \max\{2(p{-}1)/(p{-}m{+}\sum 1/\ell_i),\;\ell_1\}\cdot n$.
Our algorithm achieves
$$T \xrightarrow{k\to\infty} \max\left\{\frac{2(p{-}1)}{p{-}m},\;
\ell_1 + \frac{2(m-1)}{p-m}\right\}\cdot n.$$
In both cases, the ratio of the achieved time to the lower bound approaches~$1$ as $p \to \infty$ with $m$ fixed (or as $m/p \to 0$), confirming near-optimality in large clusters.


\section{Algorithm Design for Multi-GPU Servers}
\label{app:multi-gpu}

As outlined in Section~\ref{subsec:multi-gpu}, we now consider the practical setting where each server hosts multiple GPUs. Let $p$ be the total number of GPUs, $g$ the number of GPUs per server (typically $g \in \{1,2,4,8\}$), and $q = p/g$ the total number of servers. Each GPU within a server has its own dedicated NIC for inter-server communication, and the $g$ GPUs are additionally connected to each other through a high-bandwidth NVLink fabric.

\paragraph{Notation.}
We adopt the same symbols as Section~\ref{sec:algorithm}. The total
input of $n$ elements is first split into $g$ \emph{parts} of size
$n/g$, one per leading-GPU position; the cycle anchored at each
leading GPU operates on its own part. Each part is then split into
$k$ \emph{segments}, and each segment into $q{-}1$ \emph{sections}
of size
\[
  s \;=\; \frac{n}{g\,k\,(q{-}1)}
\]
elements. A healthy NIC therefore takes $s$ time units to transmit
one section, and an NVLink hop on one local ring (running at
$(g{-}1)\times$ NIC speed) takes $s/(g{-}1)$ time. The pattern-overlay
construction of Section~\ref{subsec:multi-gpu-pipeline} requires
$8 \mid k$ (so that the four patterns each receive an integer
number of segments) and $k \ge 24$ (so that the head, the
steady-state interior, and the tail of the schedule do not
overlap); we assume both throughout.  Throughout this appendix we adopt the
following bandwidth model for the NVLink fabric, which reflects the
configuration commonly deployed in practice:
\begin{itemize}[leftmargin = *]
  \item Each GPU has its own pair of NVLink up/down links, with
        per-direction bandwidth equal to $(g{-}1)$ times the
        NIC bandwidth, enough to keep one GPU communicating with all
        of its $g{-}1$ siblings at full NIC speed simultaneously.
  \item Different GPUs' NVLink lanes do not interfere, so all $g$
        local GPUs can send and receive at $(g{-}1)\times$ the
        NIC speed in parallel.
\end{itemize}
In particular, when one GPU broadcasts (or reduces) a value across the
remaining $g{-}1$ GPUs, the operation completes in the same time that
a single healthy NIC flow of equal size would take.  At $g=1$ the
NVLink term disappears and the model degenerates to the single-GPU
case of Section~\ref{sec:algorithm}. Thanks to PXN~\cite{nccl_pxn2022} (Section~\ref{sec:background}),
all $g$ NICs on a degraded server pool their bandwidth equally,
so we may assume without loss of generality that they
share the same slowdown factor $\ell$.

We extend our algorithm to this setting in a way that closely
parallels the classical AllReduce with NVLink: the data circulates
over NVLink inside each server in addition to crossing the
inter-server NIC ring.  Our scheduler then overlaps the NVLink work
with the NIC work so that the intra-server traffic does not extend
the critical path.

The cleanest mental model is to run $g$ concurrent copies of the
single-GPU schedule, one per local GPU index. Inside every server the
$g$ GPUs are connected as an NVLink ring; each ring instance is
``led'' by a different local GPU position, and the $g$ rings operate
in parallel without interfering with the inter-server schedule. The
leading GPU of each ring uses that ring to collect its server's data
on NVLink before each NIC send and to disperse the received data on
NVLink after each NIC receive. Figure~\ref{fig:nvlink-route} sketches
the topology: each server holds $g$ GPUs that talk to each other over
NVLink, and each GPU uses its own NIC to communicate with the
corresponding GPU on every other server, so the leading GPUs across
servers form $g$ independent inter-server rings.

\begin{figure*}[!htbp]
  \centering
  \includegraphics[width=\textwidth]{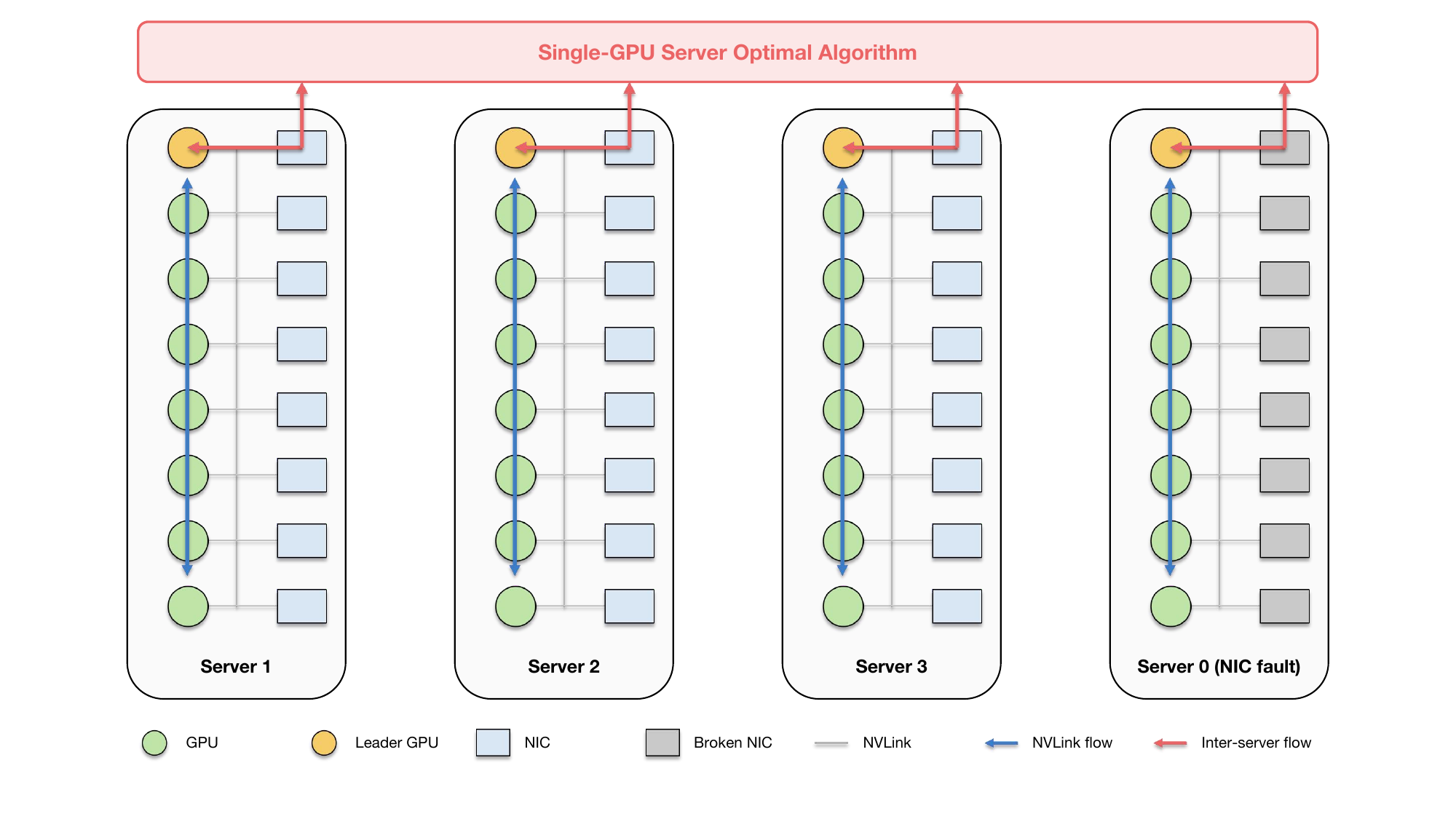}
  \caption{Multi-GPU servers: GPUs inside the same server communicate
    over NVLink, while GPUs on different servers communicate over the
    NIC. Each server hosts $g$ GPUs and each GPU has its own NIC for inter-server traffic.  We
    run $g$ instances of the single-GPU schedule in parallel, anchored
    at each local GPU index in turn; intra-server NVLink transfers
    aggregate data before each NIC send and disperse it after each
    NIC receive.  Each NVLink lane is provisioned at $(g{-}1)\times$
    NIC bandwidth and operates independently per GPU and per
    direction.}
  \label{fig:nvlink-route}
\end{figure*}

\subsection{NVLink Phases Around the Four Stages}
\label{subsec:nvlink-phases}

Each of the four stages of Section~\ref{subsec:four-stage} now requires
NVLink work in addition to its NIC traffic. Before a section is sent
out over the NIC, the $g$ local GPUs must \emph{collect} it on NVLink
so the leading GPU holds the value to be transmitted; after a section
arrives on the NIC, the leading GPU must \emph{distribute} it to the
other $g{-}1$ GPUs over NVLink. We rename the four original
inter-server stages \textbf{S1, S2, S3, S4} and label the new intra-server stages
\textbf{N1, N2, N3, N4}: each S$k$ is paired with a matching N$k$
placed either immediately before or immediately after it, and the
linear dependencies among these N and S stages admit the two possible
orderings shown in Figure~\ref{fig:ns-schematic}. An N stage uses
only NVLink and an S stage uses only the NIC.

\begin{figure}[!htbp]
  \centering
  \includegraphics[width=0.75\columnwidth]{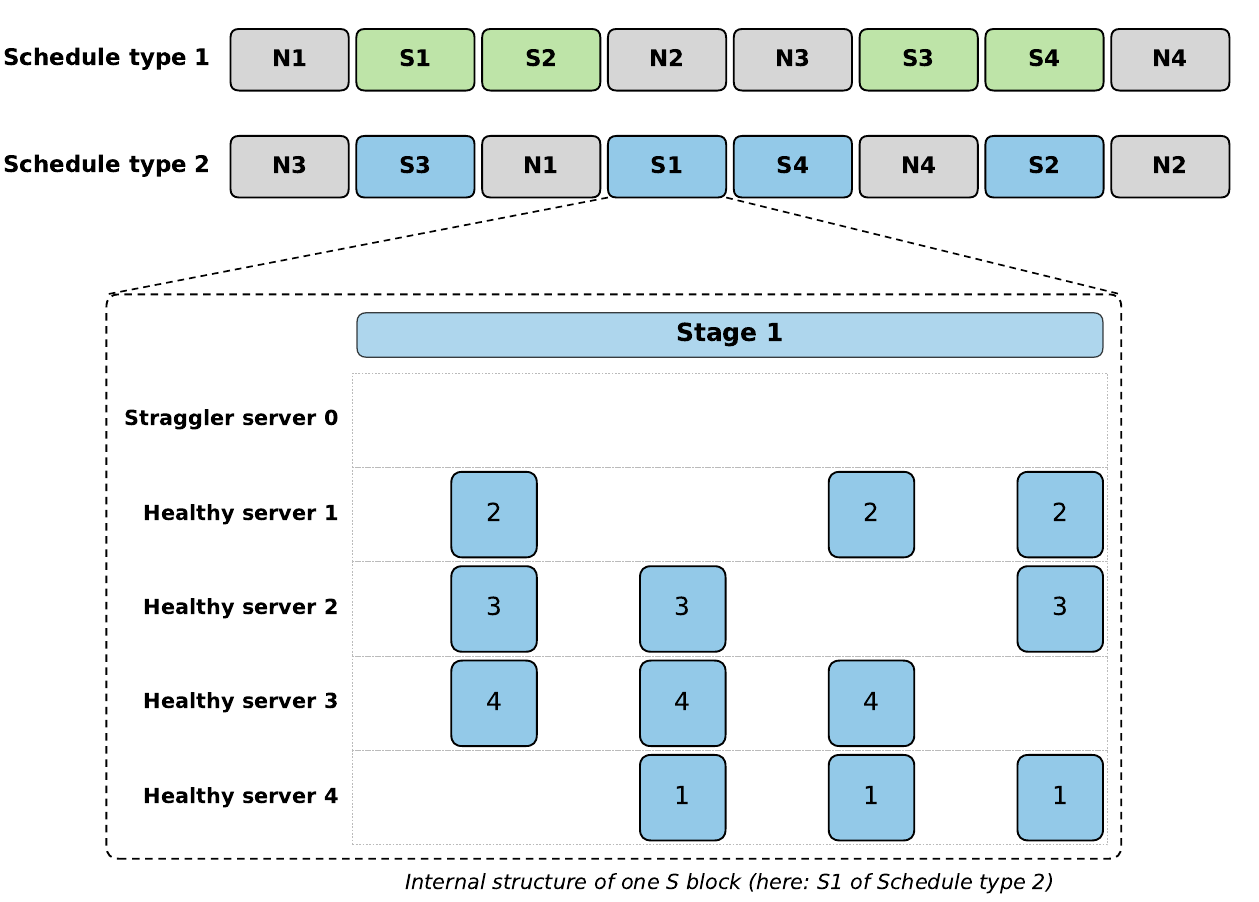}
  \caption{Schedule schematic showing the two possible orderings of
    NVLink~(N) and NIC~(S) stages that augment the four-stage
    decomposition of Section~\ref{subsec:four-stage}. The two
    schedule types correspond to the two stage orderings
    discussed therein.}
  \label{fig:ns-schematic}
\end{figure}

The role of each N phase mirrors the corresponding S phase, and the
ordering N $\leftrightarrow$ S is fixed by data dependencies:
\begin{itemize}[leftmargin = *]
  \item \textbf{N1} (collect): before S1, each healthy server
        reduces its local segment over NVLink onto the leading GPU,
        so that the leading GPU holds the value about to be sent over
        the NIC.
  \item \textbf{N2} (distribute): after S2, the straggler's
        leading GPU passes the just-received partial sum over NVLink
        to the other $g{-}1$ local GPUs. Under Schedule type~2 this
        is a plain broadcast; under Schedule type~1 the data instead
        traverses the local NVLink ring in one direction, with each
        GPU folding in its own contribution as it passes through, so
        that the GPU at the end of the ring accumulates the full
        partial sum.
  \item \textbf{N3} (collect): before S3, the straggler's
        leading GPU obtains over NVLink the value about to be sent
        over the NIC. Under Schedule type~2 this is a plain reduce of
        the local copies onto the leader; under Schedule type~1 the
        accumulated partial sum left at the end of the N2 ring
        traverses the same ring in the reverse direction, depositing
        the value onto every local GPU along the way before arriving
        at the leading GPU.
  \item \textbf{N4} (distribute): after S4, each healthy
        server's leading GPU broadcasts the just-received global sum
        over NVLink to the other $g{-}1$ local GPUs.
\end{itemize}
Data dependencies constrain the order
of the eight phases within one loop iteration;
Figure~\ref{fig:ns-schematic} shows two viable orderings.

\subsection{Alternating N and S}
\label{subsec:multi-gpu-pipeline}

Schedule types in Figure~\ref{fig:ns-schematic} do not parallelise
cleanly across the four patterns: in some places two S phases sit
next to each other while in others a single S is separated from the
next S by an N, so no cyclic shift of the eight-phase sequence can
align the four patterns into columns containing four distinct S
phases or four N phases. To obtain such an alignment we
defer selected phases to a later loop iteration, carrying along every
phase that depends on them so that the relative order required by the
data dependencies is preserved.

This gives the strict \emph{N--S alternation} of
Figure~\ref{fig:ns-8patterns}. Patterns~A and~C use the
Schedule-type-1 sequence and patterns~B and~D use the
Schedule-type-2 sequence. The figure has eight rows: rows 1--4 list
patterns A, B, C, D once each, and rows 5--8 repeat the same four
patterns shifted by one column, so that the eight rows together fill
every NIC and every NVLink slot. Reading only the S phases of the
first four rows recovers the order and overlay of the four
single-GPU patterns in Figure~\ref{fig:patterns-combined}. A primed
label (e.g., N3$'$) marks a phase whose data comes from the
\emph{previous} loop iteration; a doubly-primed label (e.g.,
N4$''$) marks data from two iterations earlier.

\begin{figure}[!htbp]
  \centering
  \begin{minipage}[b]{0.48\textwidth}
    \centering
    \includegraphics[width=\textwidth]{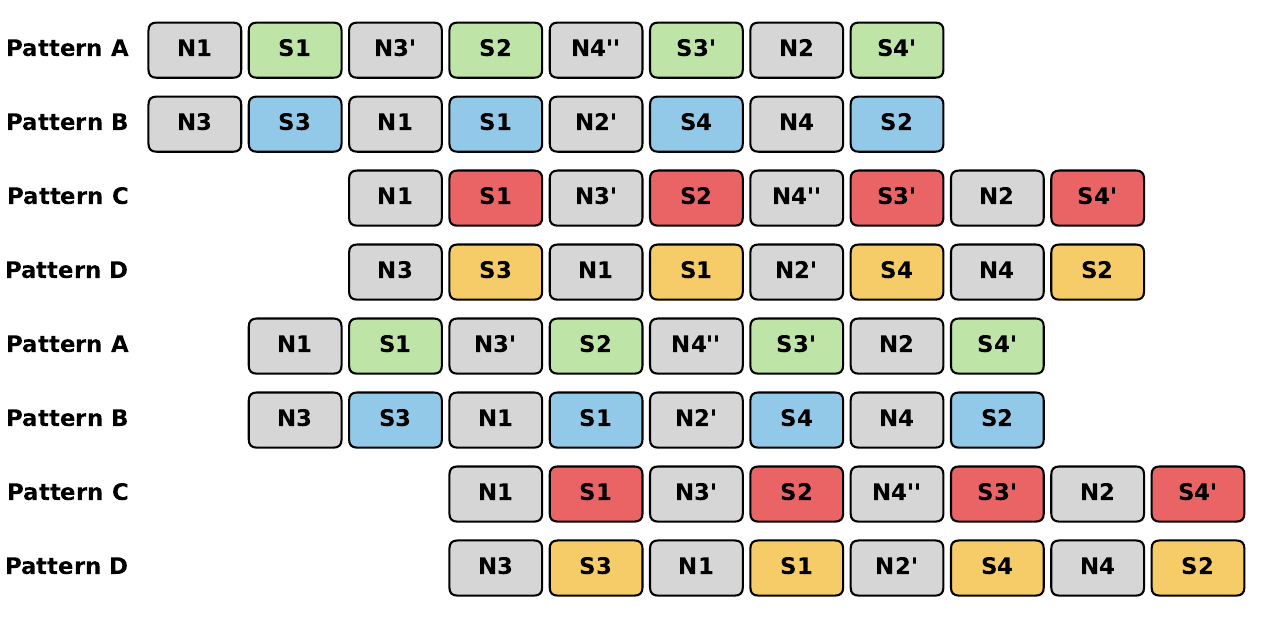}
    \captionof{figure}{Eight-row N--S schedule built from schedule types 
      of Figure~\ref{fig:ns-schematic}. A primed
      label (e.g., N3$'$) marks a phase whose data comes from the
      previous iteration; a doubly-primed label (e.g., N4$''$)
      marks data from two iterations earlier.}
    \label{fig:ns-8patterns}
  \end{minipage}\hfill
  \begin{minipage}[b]{0.48\textwidth}
    \centering
    \includegraphics[width=\textwidth]{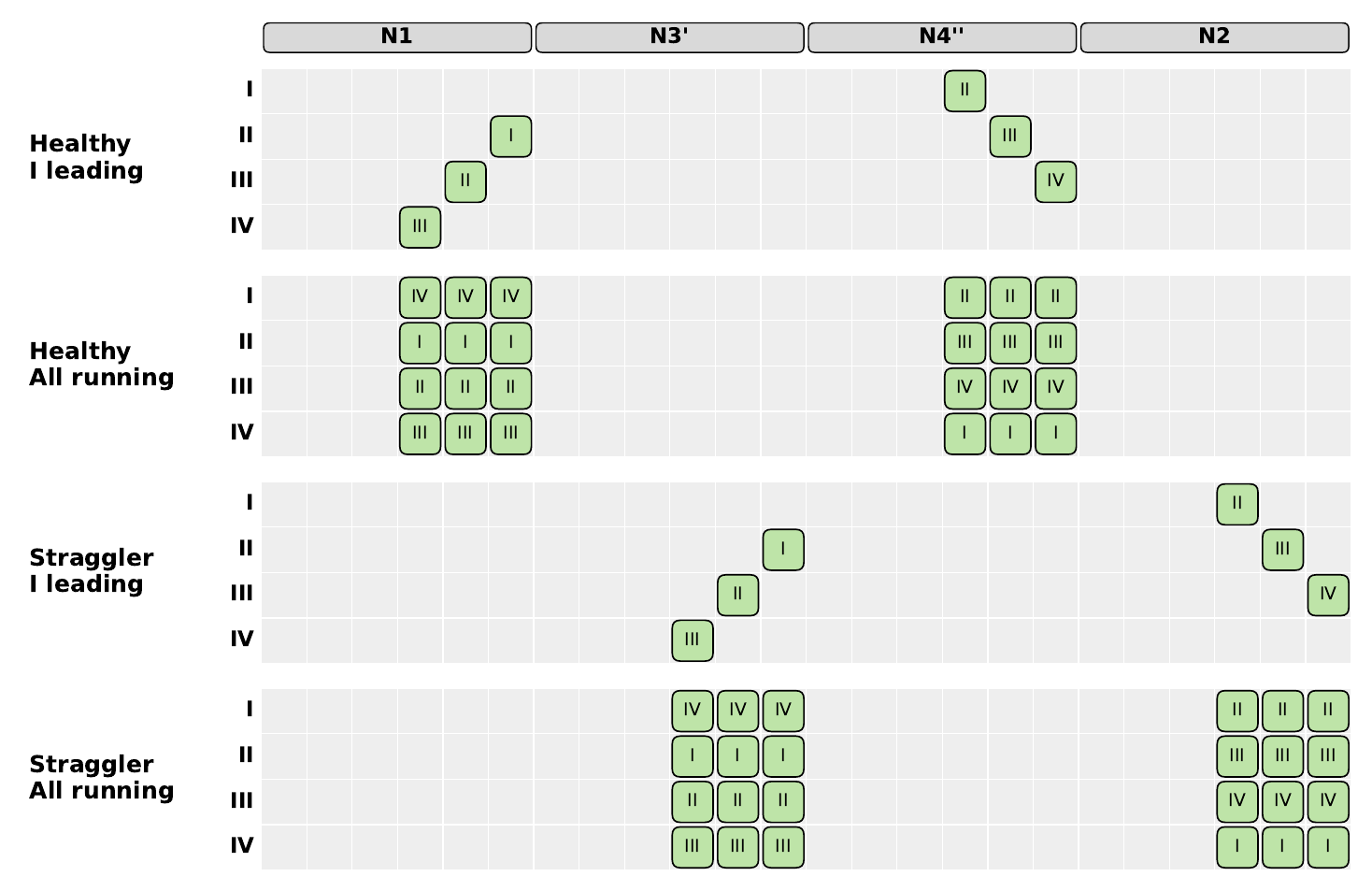}
    \captionof{figure}{NVLink data flow, shown for
      Pattern~A (other patterns follow the same template). Rows
      index the local GPUs on a server ($g=4$). Each cell is
      one NVLink hop, labelled by the destination GPU; one hop
      carries $(p{-}1)\,s$ elements.}
    \label{fig:nvlink-pattern-a}
  \end{minipage}
\end{figure}

Inside a single N body the intra-server work itself runs as a small
NVLink ring across the $g$ local GPUs.
Figure~\ref{fig:nvlink-pattern-a} illustrates the traffic for
Pattern~A: the healthy and straggler servers each contribute an
$g$-row panel, and each cell labels the destination GPU on the local
ring. (We set $g=4$ in the figure for illustration and label the four
GPUs I/II/III/IV; the construction generalises to arbitrary~$g$.)
The collect phases (N1, N3) and the distribute phases (N2, N4)
run the NVLink ring in opposite directions: collect uses the
backward ring
$\mathrm{IV}\to\mathrm{III}\to\mathrm{II}\to\mathrm{I}$ (leader
receives, never sends), and distribute uses the forward ring
$\mathrm{I}\to\mathrm{II}\to\mathrm{III}\to\mathrm{IV}$ (leader
sends first). Combining the data flows from all $g$ choices of
leading GPU exactly fills the ring: at every sub-slot every GPU is
sending to its ring neighbour, so all $g$ NVLink links are saturated
in parallel. Because NVLink is a factor of $g{-}1$ faster than the
NIC, and one pattern inside a GPU-led cycle carries
$(q{-}1)\,s = n/(gk)$ elements, the NVLink ring can cycle the data of all
$g$ leadings to every local GPU in just $s$ time units, i.e.\ half
a body, as shown in Figure~\ref{fig:nvlink-pattern-a}.

Putting the four patterns together, both transports (NIC and NVLink)
and both server types (healthy and straggler) combine into the
composite schedule of Figure~\ref{fig:ns-combined}. The Inter-server
panel is identical to the single-GPU composite of
Section~\ref{subsec:schedule} (Figure~\ref{fig:patterns-combined});
the two NVLink panels show how the local rings at each server type
are kept busy in parallel. When $\ell=2$, joining body~1 to body~9
turns Figure~\ref{fig:ns-combined} into a closed cycle in which the
NVLink bandwidth saturates in every odd body and the NIC bandwidth
saturates in every even body; overlaying a second copy of the cycle
shifted by one body then keeps both transports saturated at every
instant.

\subsection{Tail Optimization}
\label{subsec:multi-gpu-tail}

Since the primed phases cross iteration boundaries (N3$'$, S3$'$,
S4$'$ reach back one iteration and N4$''$ two), the schedule is not
yet periodic at its endpoints. At the head, the first cycle
(first $8$ bodies of each pattern as shown in Figure~\ref{fig:ns-combined}) must
leave every primed and doubly-primed phase empty and the second
cycle must leave every doubly-primed phase empty, giving a head of
$16$ bodies; one can check that no further compression is possible.
From body~$17$ onward every parallel body is filled.
At the tail, the matching leftovers --- one cycle that runs only the
primed and doubly-primed phases, plus one cycle that runs only the
doubly-primed ones --- can be packed much more tightly: by merging
remaining phases whenever the two
can share a parallel body (e.g., a body labelled ``S4/N4'' carries
both a Stage-4 NIC flow and the matching N4 step at the same time),
the entire tail fits in at most $6$ bodies, as shown in
Figure~\ref{fig:ns-end}. The last tail body is physically shorter
than the others because its remaining inter-server work --- a sole
S4 phase of duration~$s(q-2)/(q-1)$ --- takes less than half of a normal $2s$
body.

This design applies to any value of~$\ell$, and the bottleneck is
set by which resource saturates first. For $\ell>2$ the straggler
NIC is the bottleneck and the healthy NICs and both NVLink rings
still have spare capacity; for $\ell<2$ only the straggler NIC is
under-utilised, while every healthy NIC and both NVLink rings run
fully loaded.

\begin{figure*}[p]
  \centering
  \includegraphics[width=0.95\textwidth]{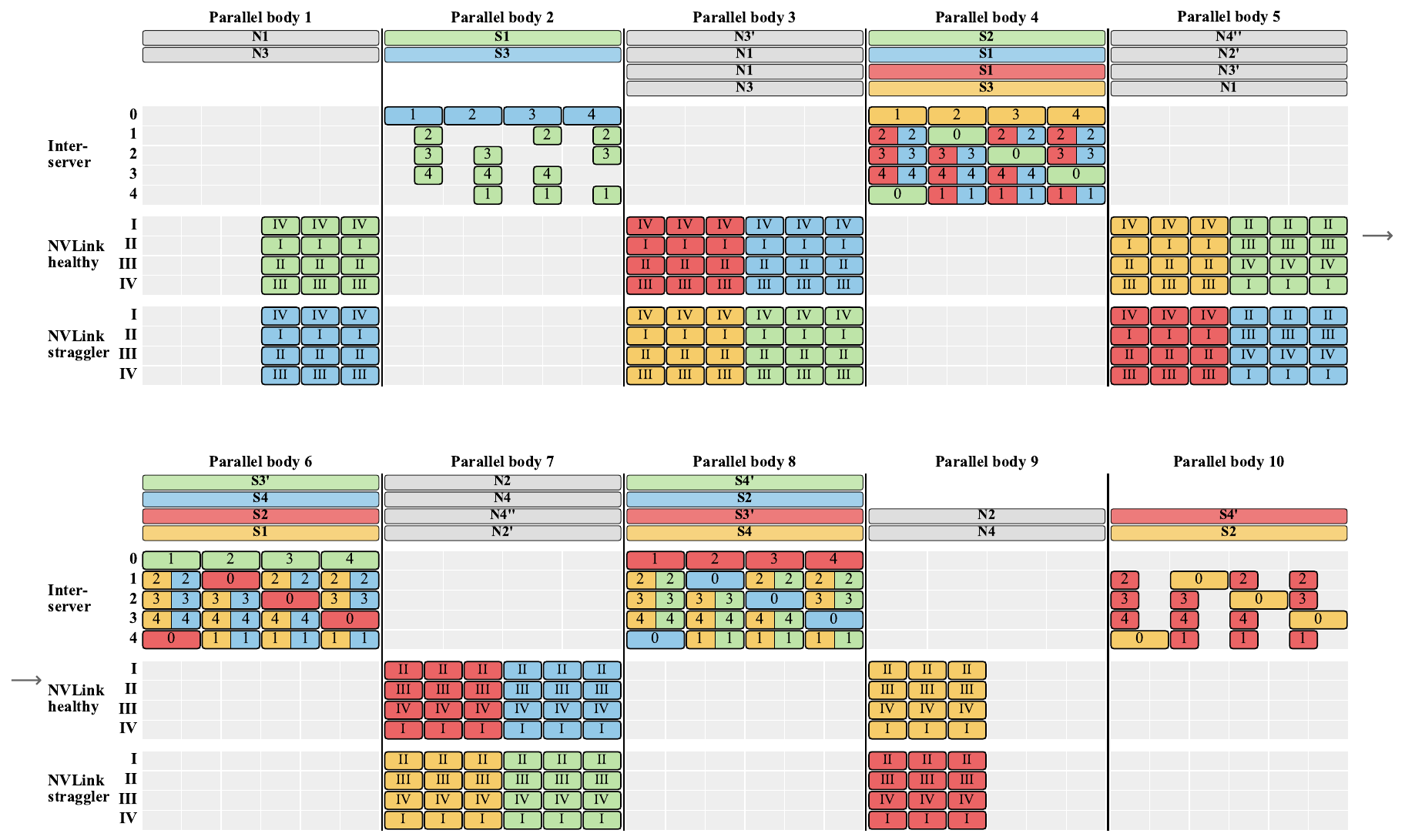}
  \caption{Composite multi-GPU schedule, drawn for $\ell=2$. Each
    parallel body is either an N body (NVLink active, NIC idle) or
    an S body (NIC active, NVLink idle); N and S alternate strictly.
    The Inter-server table matches
    Figure~\ref{fig:patterns-combined}; the NVLink healthy and
    NVLink straggler tables show how the local rings at each server
    type are filled. Closing this schedule into a cycle (body~1
    glued to body~9) and overlaying a second copy shifted by one
    body fills every parallel body.}
  \label{fig:ns-combined}
\end{figure*}

\begin{figure*}[p]
  \centering
  \includegraphics[width=0.9\textwidth]{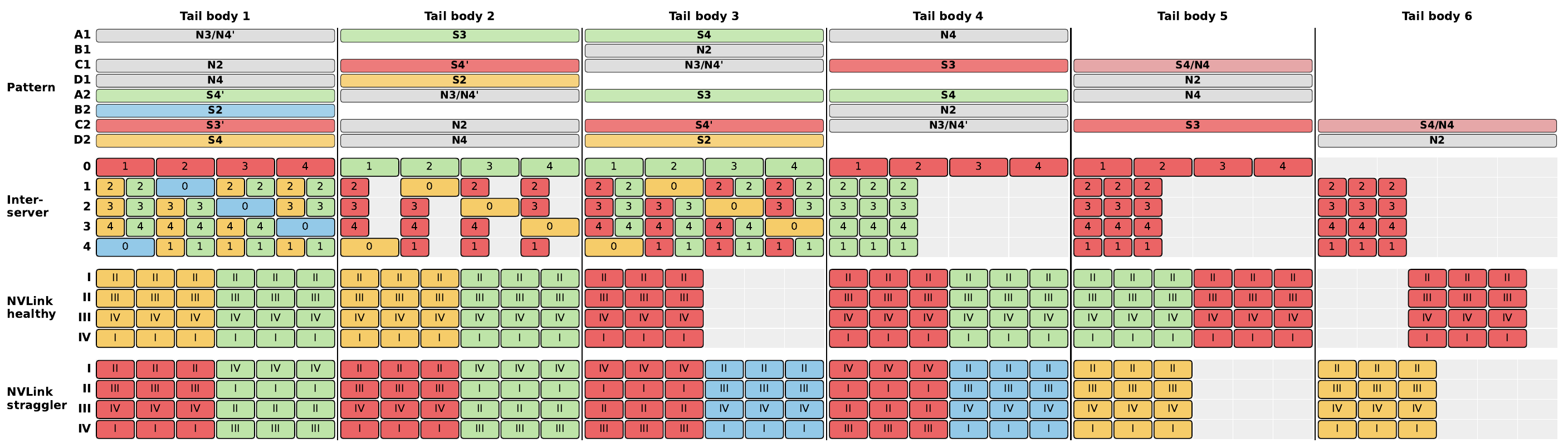}
  \caption{Tail schedule for the multi-GPU algorithm, drawn for
    $\ell=2$. The eight stage-label rows (A1--D2) are exactly the
    four pattern rows of Figure~\ref{fig:ns-combined} together with
    the same four rows shifted by one body --- the two copies whose
    overlay fills the steady-state schedule.
    A label of the form ``S$k$/N$k$'' marks a tail body in which two phases are fused.}
  \label{fig:ns-end}
\end{figure*}

\subsection{Time Analysis}
\label{subsec:multi-gpu-time}

Recall that one healthy NIC hop on a section of $s$ elements takes
$s$ time units. Under our bandwidth model each NVLink lane runs at
$(g{-}1)\times$ NIC speed, so a single NVLink hop on one local ring
takes $s/(g{-}1)$ time. Each N body actually performs \emph{two}
ring traversals --- one collect and one distribute --- for a total
of $2(g{-}1)$ hops, costing $2(g{-}1)\cdot s/(g{-}1)=2s$. For
$\ell\le 2$ this matches the $2s$ duration of the S body exactly,
and for $\ell>2$ it is strictly shorter than the $\ell s$-long S
body, so within the steady-state cycle the NVLink work always fits
and adds no new critical path; the head and tail need a separate
accounting.

The remaining accounting then follows the single-GPU analysis of
Section~\ref{subsec:time-analysis} \emph{without} the bubble-filling
optimization; the bubble-filling subtlety is addressed in
Case~$\ell<2$ below.

Counting bodies is now straightforward (using the head/tail
decomposition of Section~\ref{subsec:multi-gpu-tail}). With $k$
segments in total, the steady-state interior runs from body~$17$
through body~$k$, for $k-16$ fully filled bodies; the head spans
the first $16$ bodies but is effectively only $15.5$ body-lengths
since the first half of body~$1$ is empty; the tail occupies at
most $6$ bodies. The total schedule length is therefore at most
\[
  (k - 16) \;+\; 15.5 \;+\; 6 \;=\; k + 5.5
\]
body-lengths.

\paragraph{Case $\ell \ge 2$ (straggler-bottlenecked).}
Each parallel body takes $\ell\,(q{-}1)\,s = \ell n/(gk)$ time.
Multiplying by the body count $k+5.5$ derived above,
\begin{equation}
\label{eq:multi-gpu-large-l}
  T \;\le\; \ell\,(q{-}1)\,s\,(k + 5.5)
       \;=\; \frac{\ell\,n\,(k+5.5)}{g\,k}
       \;\xrightarrow{k\to\infty}\;
       \frac{\ell\,n}{g}\,,
\end{equation}
i.e.\ the $g$ leading-GPU cycles, each handling $n/g$ elements via
its own ring of NICs, jointly recover a $g\times$ speed-up over
the single-GPU schedule of
Equation~\eqref{eq:total-time-large-l}. This matches the
$\ell n/g$ lower bound from
Theorem~\ref{thm:multigpu-lb} and Theorem~\ref{thm:multigpu-tight-lb} in this regime, so the
algorithm is \emph{bandwidth optimal without requiring $p \rightarrow \infty$}.

\paragraph{Case $\ell < 2$ (ring-bottlenecked).}
Without bubble filling each body takes $2\,(q{-}1)\,s = 2n/(gk)$
time regardless of~$\ell$, giving
\begin{equation}
\label{eq:multi-gpu-small-l}
  T \;\le\; 2\,(q{-}1)\,s\,(k + 5.5)
       \;=\; \frac{2\,n\,(k+5.5)}{g\,k}
       \;\xrightarrow{k\to\infty}\;
       \frac{2\,n}{g}\, .
\end{equation}
Recall that Theorem~\ref{thm:multigpu-tight-lb} (strictly tighter than Theorem~\ref{thm:multigpu-lb}) gives the
lower bound 
$$T_{\mathrm{LB}} = \frac{n}{g}\cdot\frac{2\ell (q-1)}{\ell (q-2) + 2}$$ for
$\ell < 2$. The ratio of our $T = 2n/g$ to this bound
is
\[
\frac{T}{T_{\mathrm{LB}}}
   \;=\; \frac{\ell (q-2) + 2}{\ell (q-1)} \;<\; \frac{q}{q-1}\,.
\]
so the algorithm is within a factor $q/(q{-}1)$ of the theoretical
lower bound across the whole $\ell<2$ range. As $q \to \infty$, this factor approaches $1$, confirming near-optimality in large clusters.

\paragraph{Bubble-filling in the multi-GPU server case.} The bubble-filling technique of Section~\ref{subsec:schedule} cannot be directly applied here: for $\ell \leq 2$, each body already saturates the local NVLink ring, leaving no spare NVLink bandwidth to carry the additional P2P traffic that bubble filling would inject. In principle, one could slice each straggler bubble into $g$ pieces and distribute them across the $g$ leading-GPU cycles to recover that bandwidth. This would close the remaining $q/(q{-}1)$ gap and match the lower bound of Theorem~\ref{thm:multigpu-tight-lb} exactly without requiring $q \rightarrow \infty$ (following a time analysis similar to Appendix~\ref{app:bubble-analysis}); however, the resulting savings are marginal and do not justify the added scheduling complexity.

\subsection{Potential Extension to Multiple Straggler Multiple GPU per Server Case}

The alternating N-S phase construction here can be combined with the new S phase developed in Appendix~\ref{app:multi-straggler-algo}. The only obstacle is that patterns A and C are now replaced by patterns A' and C', which also alters the alternating N-S schedule. A modified schedule is provided in Figure~\ref{fig:ns-8patterns-variant}. Due to space constraints, we omit the full construction details, time analysis, and corresponding theoretical lower bound.

\begin{figure}[!htbp]
\centering
\includegraphics[width=0.5\textwidth]{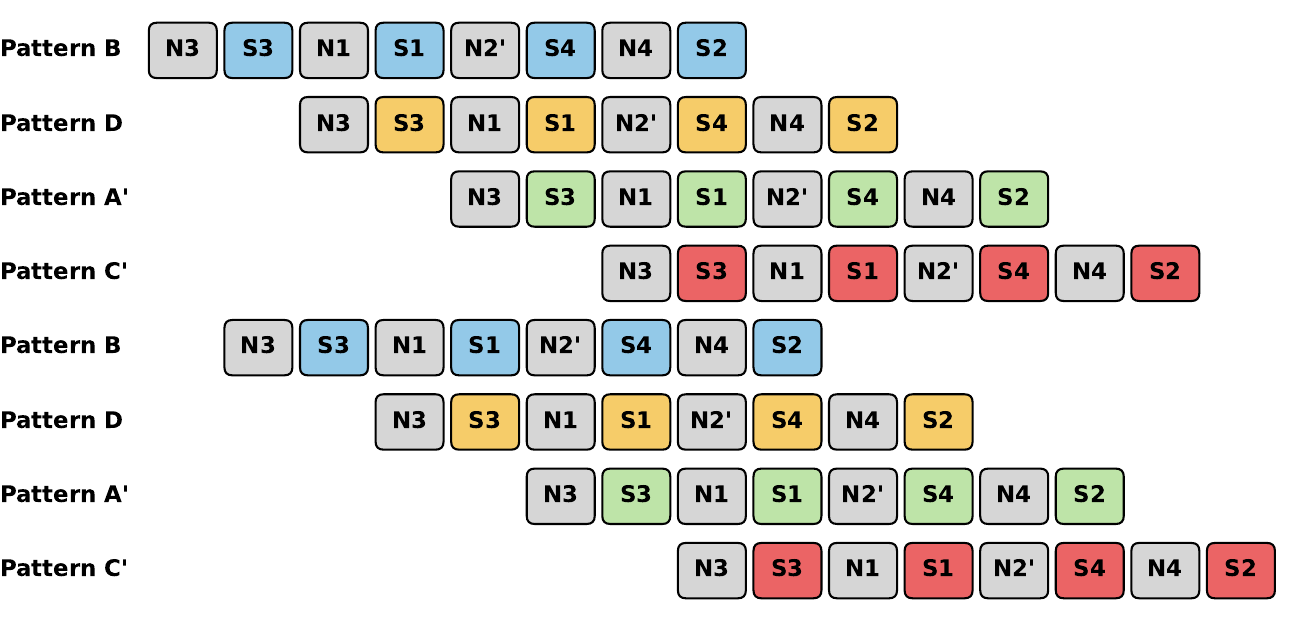}
\caption{Eight-row alternating N--S schedule with patterns A$'$ and C$'$ substituted for patterns A and C.}
\label{fig:ns-8patterns-variant}
\end{figure}

\section{Large-Message Evaluation (\texorpdfstring{$N{=}8\,$}{N=8}GiB)}
\label{app:large-n-eval}

This appendix repeats the three sweeps of
Section~\ref{sec:experiments} at the much larger message size
$N{=}8$\,GiB --- well beyond the buffer sizes that real DDP/FSDP
implementations use, but useful as a stress test that exposes the
asymptotic per-byte rate of each algorithm. Single-straggler results are in
Figure~\ref{fig:eval-curves}, multi-straggler in
Figure~\ref{fig:eval-multi-app}, and multi-GPU-per-server in
Figure~\ref{fig:eval-mgpu-app}. All trends discussed in the main text
(OptCC tracking NCCL\textsubscript{NoFailure} within
$1{+}1/(p{-}1)$ overhead, linear scaling in $N$, and the
$\ell\le 2$ plateau) persist at this larger message size.

\begin{figure*}[h]
  \centering
  \begin{subfigure}[b]{0.32\textwidth}
    \centering
    \includegraphics[width=\textwidth]{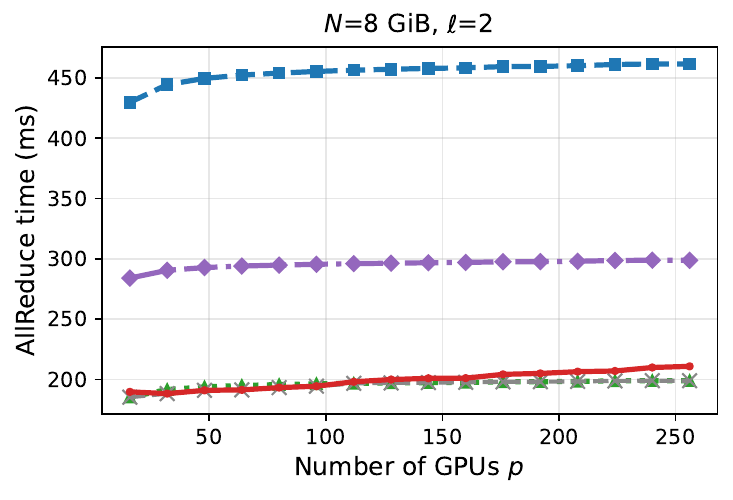}
    \caption{Varying $p$ (GPUs).}
    \label{fig:eval-p}
  \end{subfigure}\hfill
  \begin{subfigure}[b]{0.32\textwidth}
    \centering
    \includegraphics[width=\textwidth]{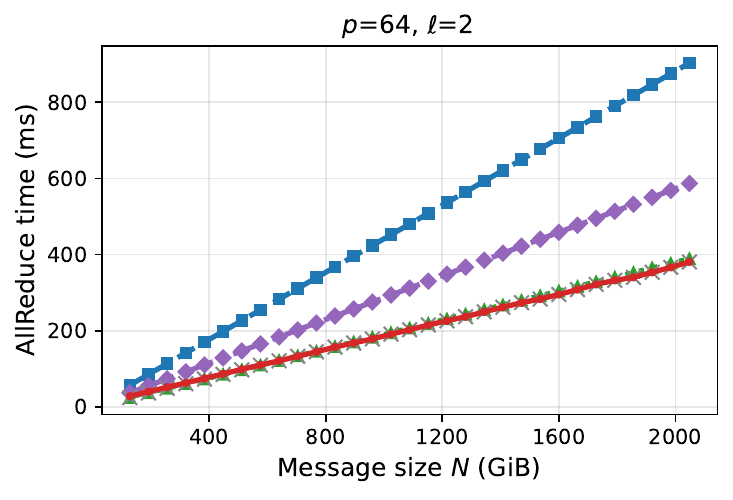}
    \caption{Varying $n$ (message size).}
    \label{fig:eval-n}
  \end{subfigure}\hfill
  \begin{subfigure}[b]{0.32\textwidth}
    \centering
    \includegraphics[width=\textwidth]{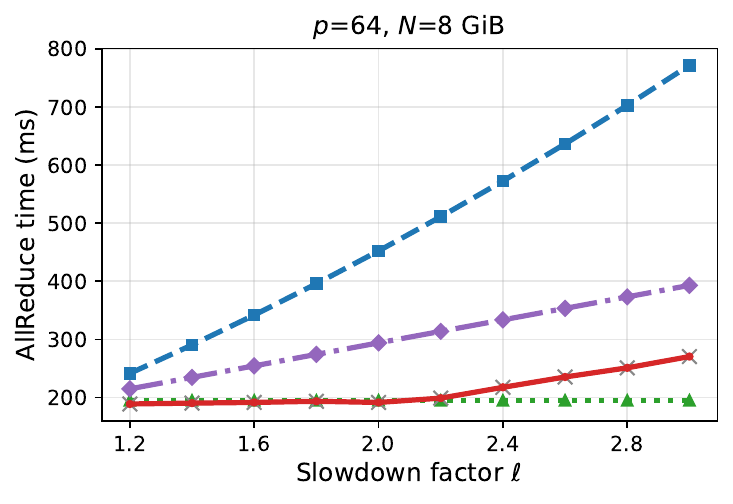}
    \caption{Varying $\ell$ (slowdown).}
    \label{fig:eval-l}
  \end{subfigure}
  \caption{Single-straggler AllReduce completion time on SimAI for
    OptCC versus ICCL, NCCL\textsubscript{NoFailure} (no failure), and
    R$^2$CCL, at $N{=}8$\,GiB. Filled markers are SimAI measurements;
    open markers are estimated values for bins not yet measured
    (filled in via the analytical model with empirical scaling).
    R$^2$CCL is plotted from its closed form
    $T_{R^2CCL}=T_{\text{NCCL}_\text{optimal}}\,\bigl(1+\tfrac{p(\ell-1)}{2(p-1)}\bigr)$;
    no measurements yet. \textbf{(a)} As $p$ grows, OptCC's completion
    time approaches NCCL\textsubscript{NoFailure} while the gap to NCCL and
    R$^2$CCL widens. \textbf{(b)} The absolute time gap grows linearly
    with message size~$n$. \textbf{(c)} As the straggler becomes slower
    ($\ell$ increases), OptCC remains bounded by $\max\{2,\ell\}\,n$
    whereas NCCL and R$^2$CCL grow linearly in~$\ell$ from $\ell{=}1$.}
  \label{fig:eval-curves}
\end{figure*}

\begin{figure*}[h]
  \centering
  \begin{subfigure}[b]{0.32\textwidth}
    \centering
    \includegraphics[width=\textwidth]{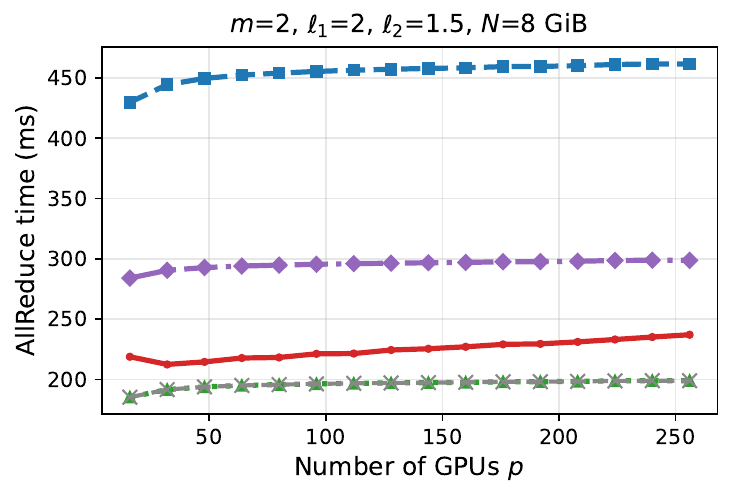}
    \caption{Varying $p$ (GPUs).}
    \label{fig:eval-multi-p}
  \end{subfigure}\hfill
  \begin{subfigure}[b]{0.32\textwidth}
    \centering
    \includegraphics[width=\textwidth]{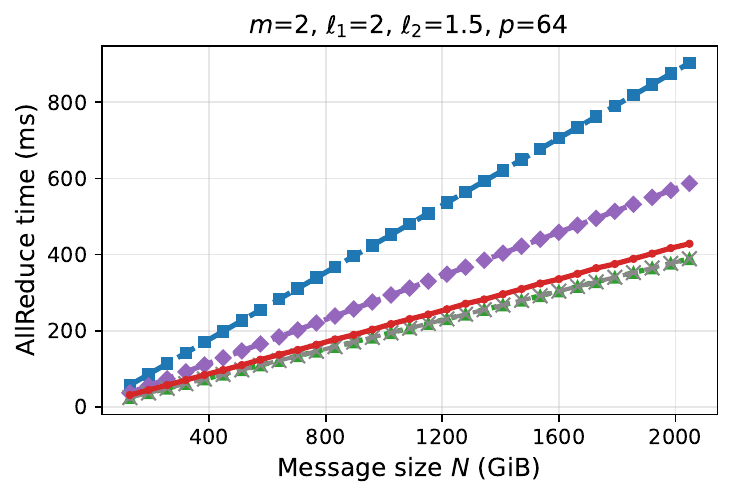}
    \caption{Varying $N$ (message size).}
    \label{fig:eval-multi-n}
  \end{subfigure}\hfill
  \begin{subfigure}[b]{0.32\textwidth}
    \centering
    \includegraphics[width=\textwidth]{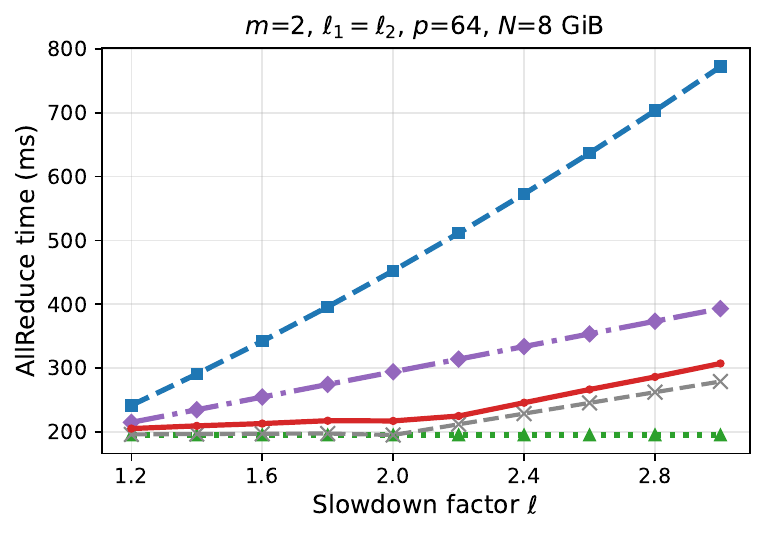}
    \caption{Varying $\ell$ ($\ell_1{=}\ell_2$).}
    \label{fig:eval-multi-l}
  \end{subfigure}
  \caption{Multi-straggler AllReduce ($m{=}2$) at $N{=}8$\,GiB.
    OptCC vs.\ NCCL\textsubscript{NoFailure} (no failure).
    \textbf{(a)}~$\ell_1{=}1.5$, $\ell_2{=}2$, $N{=}8$\,GiB;
    \textbf{(b)}~$p{=}64$, $\ell_1{=}1.5$, $\ell_2{=}2$;
    \textbf{(c)}~$p{=}64$, $N{=}8$\,GiB, $\ell_1{=}\ell_2$.}
  \label{fig:eval-multi-app}
\end{figure*}

\begin{figure*}[h]
  \centering
  \begin{subfigure}[b]{0.32\textwidth}
    \centering
    \includegraphics[width=\textwidth]{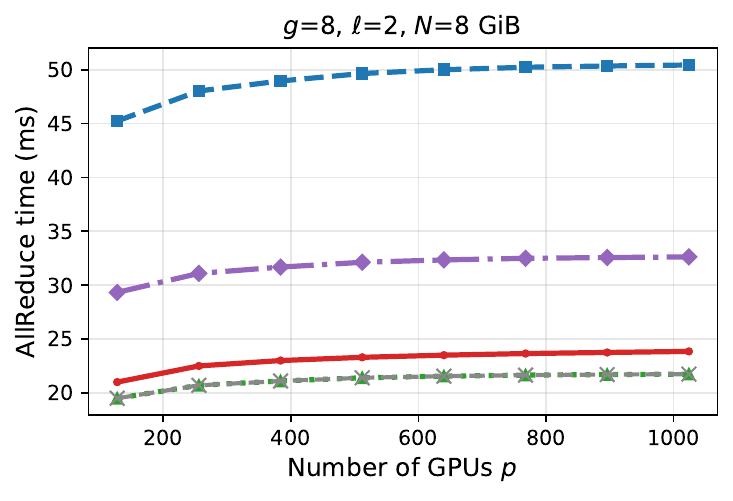}
    \caption{Varying $p$ (GPUs).}
    \label{fig:eval-mgpu-p}
  \end{subfigure}\hfill
  \begin{subfigure}[b]{0.32\textwidth}
    \centering
    \includegraphics[width=\textwidth]{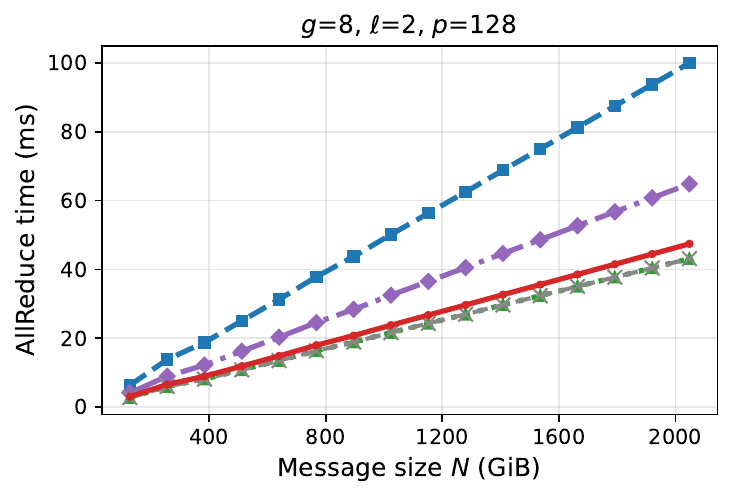}
    \caption{Varying $N$ (message size).}
    \label{fig:eval-mgpu-n}
  \end{subfigure}\hfill
  \begin{subfigure}[b]{0.32\textwidth}
    \centering
    \includegraphics[width=\textwidth]{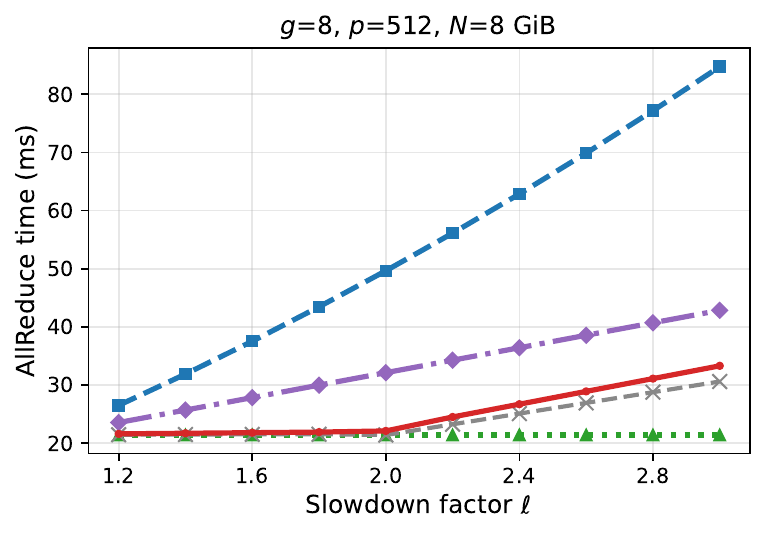}
    \caption{Varying $\ell$ (slowdown).}
    \label{fig:eval-mgpu-l}
  \end{subfigure}
  \caption{Multi-GPU-per-server AllReduce ($g{=}8$) at $N{=}8$\,GiB.
    OptCC vs.\ ICCL, NCCL\textsubscript{NoFailure}, and R$^2$CCL.
    \textbf{(a)}~$\ell{=}2$, $N{=}8$\,GiB;
    \textbf{(b)}~$p{=}64$, $\ell{=}2$;
    \textbf{(c)}~$p{=}64$, $N{=}8$\,GiB.}
  \label{fig:eval-mgpu-app}
\end{figure*}

\end{document}